\documentclass[11pt,a4paper]{amsart}
\usepackage{amsmath,amsfonts,amsthm,amssymb}
\usepackage[ruled]{algorithm2e}
\usepackage{color}
\usepackage{epsfig}
\usepackage{graphicx}
\usepackage{cite,url}
\usepackage{hyperref}
\usepackage{xspace}
\usepackage{enumerate}
\usepackage[left=1in,right=1in,top=1in,bottom=1in,foot=0.5in]{geometry}

\usepackage[leqno]{amsmath}
\usepackage{graphicx}
\usepackage[labelfont=rm]{subcaption}
\usepackage{cite,url}
\usepackage{hyperref}
\usepackage[left=1in,right=1in,top=1in,bottom=1in,foot=0.5in]{geometry}
\usepackage{enumerate}
\usepackage{color}
\usepackage{tabularx}

\theoremstyle{plain}
\newtheorem{theorem}{Theorem}[section]
\newtheorem{lemma}[theorem]{Lemma}
\newtheorem{proposition}[theorem]{Proposition}
\newtheorem{corollary}[theorem]{Corollary}

\theoremstyle{definition}

\newtheorem{remark}[theorem]{Remark}

\newtheorem{example}[theorem]{Example}

\numberwithin{equation}{section}

\DeclareMathOperator*{\Med}{Med}
\DeclareMathOperator*{\maj}{maj}

\newcommand{\lMedw}{\ensuremath{\Med^{\rm{loc}}_w}\xspace}
\newcommand{\G}{\ensuremath{\mathcal{G}}\xspace}
\newcommand{\hG}{\ensuremath{\widehat{\mathcal{G}}}\xspace}
\newcommand{\cN}{\ensuremath{\mathcal{N}}\xspace}
\newcommand{\wM}{\ensuremath{\widehat{M}}\xspace}
\newcommand{\ch}{\ensuremath{\mathfrak{h}}\xspace}
\newcommand{\cH}{\ensuremath{\mathcal{H}}\xspace}
\newcommand{\hB}{\ensuremath{\widehat{Q}}\xspace}
\newcommand{\hE}{\ensuremath{\widehat{E}}\xspace}
\newcommand{\hepsilon}{\ensuremath{\widehat{\epsilon}}\xspace}

\newcommand{\wG}{\ensuremath{\widehat{G}}}
\newcommand{\hm}{\ensuremath{\widehat{m}}}
\newcommand{\tG}{\ensuremath{\widetilde{G}}}
\newcommand{\tE}{\ensuremath{\widetilde{E}}}
\newcommand{\tH}{\ensuremath{\widetilde{H}}}
\newcommand{\tw}{\ensuremath{\widetilde{w}}}
\newcommand{\oG}{\ensuremath{\overrightarrow{G}}}

\newcommand{\owG}{\ensuremath{\overrightarrow{\widehat{G}}}}

\newcommand{\uw}{\ensuremath{{w}_*}\xspace}

\definecolor{darkgray}{rgb}{0.33, 0.33, 0.33}

\newcommand{\R}{\ensuremath{\mathbb{R}}\xspace}
\newcommand{\N}{\ensuremath{\mathbb{N}}\xspace}

\newcommand\E{\mathcal{E}}
\newcommand\cD{\mathcal{D}}

\title[Medians in median graphs and their cube complexes in linear time]{Medians in median graphs and their cube complexes in linear time$^1$}\thanks{$^1$An extended abstract of
  this paper appeared in the proceedings of ICALP 2020.}

\author[L. B\'en\'eteau]{Laurine B\'en\'eteau}
\address{Laboratoire d'Informatique et Syst\`emes, Aix-Marseille Universit\'e and CNRS}
\email{laurine.beneteau@lis-lab.fr}

\author[J.\ Chalopin]{J\' er\'emie Chalopin}
\address{Laboratoire d'Informatique et Syst\`emes, Aix-Marseille Universit\'e and CNRS}
\email{jeremie.chalopin@lis-lab.fr}

\author[V.\ Chepoi]{Victor Chepoi}
\address{Laboratoire d'Informatique et Syst\`emes, Aix-Marseille Universit\'e and CNRS}
\email{victor.chepoi@lis-lab.fr}

\author[Y.\ Vax\`es]{Yann Vax\`es}
\address{Laboratoire d'Informatique et Syst\`emes, Aix-Marseille Universit\'e and CNRS}
\email{yann.vaxes@lis-lab.fr}

\begin{document}

\begin{abstract}
  The median of a set of vertices $P$ of a graph $G$ is the set of all
  vertices $x$ of $G$ minimizing the sum of distances from $x$ to all
  vertices of $P$.  In this paper, we present a linear time algorithm
  to compute medians in median graphs, improving over the existing
  quadratic time algorithm. We also present a linear time algorithm to
  compute medians in the $\ell_1$-cube complexes associated with
  median graphs.  Median graphs constitute the principal class of
  graphs investigated in metric graph theory and have a rich geometric
  and combinatorial structure, due to their bijections with CAT(0)
  cube complexes and domains of event structures.  Our algorithm is
  based on the majority rule characterization of medians in median
  graphs and on a fast computation of parallelism classes of edges
  ($\Theta$-classes or hyperplanes) via Lexicographic Breadth First
  Search (LexBFS).  To prove the correctness of our algorithm, we show
  that any LexBFS ordering of the vertices of $G$ satisfies the
  following \emph{fellow traveler property} of independent interest:
  the parents of any two adjacent vertices of $G$ are also
  adjacent. Using the fast computation of the $\Theta$-classes, we
  also compute the Wiener index (total distance) of $G$ in linear time
  and the distance matrix in optimal quadratic time.
\end{abstract}

\maketitle	
	
\section{Introduction}
	
The median problem (also called the Fermat-Torricelli problem or the
Weber problem) is one of the oldest optimization problems in Euclidean
geometry~\cite{LoMoWe}.  The \emph{median problem} can be defined for
any metric space $(X,d)$: given a finite set $P\subset X$ of points
with positive weights, compute the points $x$ of $X$ minimizing the
sum of the distances from $x$ to the points of $P$ multiplied by their
weights. The median problem in graphs is one of the principal models
in network location theory~\cite{Hakimi,TaFrLa} and is equivalent to
finding nodes with largest closeness centrality in network
analysis~\cite{Bav,Beau,Sabi}. It also occurs in social group choice
as the Kemeny median. In the consensus problem in social group choice,
given individual rankings of $d$ candidates one has to compute a
consensual group decision. By the classical Arrow's impossibility
theorem, there is no consensus function satisfying natural
``fairness'' axioms. It is also well-known that the majority rule
leads to Condorcet's paradox, i.e., to the existence of cycles in the
majority relation. In this respect, the Kemeny median~\cite{Ke,KeSn}
is an important consensus function and corresponds to the median
problem in graphs of permutahedra (the graph whose vertices are all
$d!$ permutations of the candidates and whose edges are the pairs of
permutations differing by adjacent transpositions).  Other classical
algorithmic problems related to distances are the diameter and center
problems. Yet another such problem comes from chemistry and consists
in the computation of the Wiener index of a graph. This is a
topological index of a molecule, defined as the sum of the distances
between all pairs of vertices in the associated chemical
graph~\cite{Wie}.  The Wiener index is closely related to the
closeness centrality of a vertex in a graph, a quantity inversely
proportional to the sum of all distances between the given vertex and
all other vertices that has been frequently used in sociometry and the
theory of social networks.
	
The median problem in Euclidean spaces cannot be solved in symbolic
form~\cite{Baj}, but can be solved numerically by Weiszfeld's
algorithm~\cite{Weis} and its convergent modifications (see e.g.\
\cite{Os}), and can be approximated in nearly linear time with
arbitrary precision~\cite{CoLeMiPaSi}.  For the $\ell_1$-metric the
median problem becomes easier and can be solved by the majority rule
on coordinates, i.e., by taking as median a point whose $i$th
coordinate is the median of the list of $i$th coordinates of the
points of $P$.  This kind of rule was used in~\cite{Jordan} to define
centroids of trees (which coincide with their
medians~\cite{GoWi,TaFrLa}) and can be viewed as an instance of the
majority rule in social choice theory.  For graphs with $n$ vertices,
$m$ edges, and standard graph distance, the median problem can be
trivially solved in $O(nm)$ time by solving the All Pairs Shortest
Paths (APSP) problem.  One may ask if APSP is necessary to compute the
median. However, by~\cite{AbGrVa} APSP and median problem are
equivalent under subcubic reductions.  It was also shown
in~\cite{AbVaWa} that computing the medians of sparse graphs in
subquadratic time refutes the HS (Hitting Set) conjecture. It was also
noted in~\cite{Cab} that computing the Wiener index of a sparse graph
in subquadratic time will refute the Exponential time (SETH)
hypothesis. Note also that computing the Kemeny median is
NP-hard~\cite{DwKuNaSi} if the input is the list of individual
preferences.
	
In this paper, we show that the medians in median graphs can be
computed in optimal $O(m)$ time (i.e., without solving APSP). Median
graphs are the graphs in which each triplet $u,v,w$ of vertices has a
unique median, i.e., a vertex metrically lying between $u$ and $v$,
$v$ and $w$, and $w$ and $u$.  They originally arise in universal
algebra~\cite{Av,BiKi} and their properties have been first
investigated in~\cite{Mu,Ne}.  It was shown in~\cite{Ch_CAT,Ro} that
the cube complexes of median graphs are exactly the CAT(0) cube
complexes, i.e., cube complexes of global non-positive
curvature. CAT(0) cube complexes, introduced and nicely characterized
in~\cite{Gromov} in a local-to-global way, are now one of the
principal objects of investigation in geometric group
theory~\cite{Sa_survey}. Median graphs also occur in Computer Science:
by~\cite{ArOwSu,BaCo} they are exactly the domains of event structures
(one of the basic abstract models of concurrency)~\cite{NiPlWi} and
median-closed subsets of hypercubes are exactly the solution sets of
2-SAT formulas~\cite{MuSch,Schaefer}. The bijections between median
graphs, CAT(0) cube complexes, and event structures have been used
in~\cite{CC-ICALP17,CC-MSO,Ch_nice} to disprove three conjectures in
concurrency~\cite{RoTh,Thi_conjecture,ThiaYa}.  Finally, median
graphs, viewed as median closures of sets of vertices of a hypercube,
contain all most parsimonious (Steiner) trees~\cite{BaFoRo} and as
such have been extensively applied in human genetics.  For a survey on
median graphs and their connections with other discrete and geometric
structures, see the books~\cite{HaImKla,Knuth}, the
surveys~\cite{BaCh_survey,KlMu_survey}, and the paper~\cite{ChChHiOs}.
	
As we noticed, median graphs have strong geometric and metric
properties.  For the median problem, the concept of $\Theta$-classes
is essential. Two edges of a median graph $G$ are called opposite if
they are opposite in a common square of $G$. The equivalence relation
$\Theta$ is the reflexive and transitive closure of this oppositeness
relation.  Each equivalence class of $\Theta$ is called a
$\Theta$-class ($\Theta$-classes correspond to hyperplanes in CAT(0)
cube complexes~\cite{Sa} and to events in event
structures~\cite{NiPlWi}).  Removing the edges of a $\Theta$-class,
the graph $G$ is split into two connected components which are convex
and gated. Thus they are called halfspaces of $G$.  The convexity of
halfspaces implies via~\cite{Dj} that any median graph $G$
isometrically embeds into a hypercube of dimension equals to the
number $q$ of $\Theta$-classes of $G$.
	
\subsection*{Our results and motivation}
In this paper, we present a linear time algorithm to compute medians
in median graphs and in associated $\ell_1$-cube complexes. Our main
motivation and technique stem from the majority rule characterization
of medians in median graphs and the unimodality of the median
function~\cite{BaBa,SoCh_Weber}. Even if the majority rule is simple
to state and is a commonly approved consensus method, its algorithmic
implementation is less trivial if one has to avoid the computation of
the distance matrix.  On the other hand, the unimodality of the median
function implies that one may find the median set by using local
search. More generally, consider a partial orientation of the input
median graph $G$, where an edge $uv$ is transformed into the arc
$\overrightarrow{uv}$ iff the median function at $u$ is larger than
the median function at $v$ (in case of equality we do not orient the
edge $uv$).  Then the median set is exactly the set of all the sinks
in this partial orientation of $G$. Therefore, it remains to compare
for every edge $uv$ the median function at $u$ and at $v$ in constant
time. For this we use the partition of the edge-set of a median graph
$G$ into $\Theta$-classes and; for every $\Theta$-class, the partition
of the vertex-set of $G$ into complementary halfspaces.  It is easy to
notice that all edges of the same $\Theta$-class are oriented in the
same way because for any such edge $uv$ the difference between the
median functions at $u$ and at $v$, respectively, can be expressed as
the sum of weights of all vertices in the same halfspace as $v$ minus
the sum of weights of all vertices in the same halfspace as $u$.
	
Our main technical contribution is a new method for computing the
$\Theta$-classes of a median graph $G$ with $n$ vertices and $m$ edges
in linear $O(m)$ time. For this, we prove that Lexicographic Breadth
First Search (LexBFS)~\cite{RoTaLu} produces an ordering of the
vertices of $G$ satisfying the following \emph{fellow traveler
  property}: for any edge $uv$, the parents of $u$ and $v$ are
adjacent.  With the $\Theta$-classes of $G$ at hand and the majority
rule for halfspaces, we can compute the weights of halfspaces of $G$
in optimal time $O(m)$, leading to an algorithm of the same complexity
for computing the median set.  We adapt our method to compute in
linear time the median of a finite set of points in the $\ell_1$-cube
complex associated with $G$. We also show that this method can be
applied to compute the Wiener index in optimal $O(m)$ time and the
distance matrix of $G$ in optimal $O(n^2)$ time.

In all previous results we assumed that the input of the problem is
given by the median graph or its cube complex, together with the set
of terminals and their weights.  However, analogously to the Kemeny
median problem, the median problem in a median graph $G$ can be
defined in a more compact way. We mentioned above that median graphs
are exactly the domains of configurations of event structures and the
solution sets of 2-SAT formulas (with no equivalent variables).  The
underlying event structure or the underlying 2-SAT formula provide a
much more compact (but implicit) description of the median graph
$G$. Therefore, we can formulate a median problem by supposing that
the input is a set of configurations of an event structure and their
weights. The goal is to compute a configuration minimizing the sum of
the weighted (Hamming) distances to the terminal-configurations.
Thanks to the majority rule, we show that this median problem can be
efficiently solved in the size of the input.  Finally, we suppose that
the input is an event structure and the goal is to compute a
configuration minimizing the sum of distances to \emph{all}
configurations.
We show that this problem is $\#$P-hard. For this we establish a
direct correspondence between event structures and 2-SAT formulas and
we use a result by Feder~\cite{Fe} that an analogous median problem
for 2-SAT formulas is $\#$P-hard.

\subsection*{Related work}
The study of the median problem in median graphs originated
in~\cite{BaBa,SoCh_Weber} and continued
in~\cite{BaBrChKlKoSu,McMuRo,Mu_Exp,MuNo,PuSl}.  Using different
techniques and extending the majority rule for trees~\cite{GoWi}, the
following \emph{majority rule} has been established
in~\cite{BaBa,SoCh_Weber}: a halfspace $H$ of a median graph $G$
contains at least one median iff $H$ contains at least one half of the
total weight of $G$; moreover, the median of $G$ coincides with the
intersection of halfspaces of $G$ containing strictly more than half
of the total weight. It was shown in~\cite{SoCh_Weber} that the median function of a median
graph is weakly convex (an analog of a discrete convex function).
This convexity property characterizes all graphs in which all local
medians are global~\cite{BaCh_median}.  A nice axiomatic
characterization of medians of median graphs via three basic axioms
has been obtained in~\cite{MuNo}. More recently, \cite{PuSl}
characterized median graphs as \emph{closed Condorcet domains},
i.e., as sets of linear orders with the property that, whenever the
preferences of all voters belong to the set, their majority relation
has no cycles and also belongs to the set. Below we will show that the
median graphs are the bipartite graphs in which the medians are
characterized by the majority rule.

Prior to our work, the best algorithm to compute the $\Theta$-classes
of a median graph $G$ has complexity $O(m\log n)$~\cite{HagImKl}.  It
was used in~\cite{HagImKl} to recognize median graphs in subquadratic
time. The previous best algorithm for the median problem in a median
graph $G$ with $n$ vertices and $q$ $\Theta$-classes has complexity
$O(qn)$~\cite{BaBrChKlKoSu} which is quadratic in the worst case.
Indeed $q$ may be linear in $n$ (as in the case of trees) and is
always at least $d(\sqrt[d]{n} - 1)$ as shown below ($d$ is the
largest dimension of a hypercube which is an induced subgraph of $G$).
Additionally, \cite{BaBrChKlKoSu} assumes that an isometric embedding
of $G$ in a $q$-hypercube is given. The description of such an
embedding has already size $O(qn)$.  The $\Theta$-classes of a median
graph $G$ correspond to the coordinates of the smallest hypercube in
which $G$ isometrically embeds (this is called the \emph{isometric
  dimension} of $G$~\cite{HaImKla}). Thus one can define
$\Theta$-classes for all partial cubes, i.e., graphs isometrically
embeddable into hypercubes.  An efficient computation (in $O(n^2)$
time) of all $\Theta$-classes was the main step of the $O(n^2)$
algorithm of~\cite{Epp} for recognizing partial cubes.  The
fellow-traveler property (which is essential in our computation of
$\Theta$-classes) is a notion coming from geometric group
theory~\cite{ECHLPT} and is a main tool to prove the (bi)automaticity
of a group.  In a slightly stronger form it allows to establish the
dismantlability of graphs (see~\cite{BrChChGoOs,Ch_dism,Ch_CAT} for
examples of classes of graphs in which a fellow traveler order was
obtained by BFS or LexBFS). LexBFS has been used to solve optimally
several algorithmic problems in different classes of graphs, in
particular for their recognition (for a survey, see~\cite{Co}).
	
Cube complexes of median graphs with $\ell_1$-metric have been
investigated in~\cite{vdV}.  The same complexes but endowed with the
$\ell_2$-metric are exactly the CAT(0) cube complexes.  As we noticed
above, they are of great importance in geometric group
theory~\cite{Sa_survey}.  The paper~\cite{BiHoVo} established that the
space of trees with a fixed set of leaves is a CAT(0) cube complex. A
polynomial-time algorithm to compute the $\ell_2$-distance between two
points in this space was proposed in~\cite{OwPr}.  This result was
recently extended in~\cite{Hayashi} to all CAT(0) cube complexes.  A
convergent numerical algorithm for the median problem in CAT(0) spaces
was given in~\cite{Bacak}.

Finally, for an extensive bibliography on Wiener index in graphs,
see~\cite{HaImKla,Kl_wiener}. The Wiener index of a tree can be
computed in linear time~\cite{MoPi}.  Using this and the fact that
benzenoids (i.e., subgraphs of the hexagonal grid bounded by a simple
curve) isometrically embed in the product of three trees,~\cite{ChKl}
proposed a linear time algorithm for the Wiener index of
benzenoids. Finally, in a recent breakthrough~\cite{Cab}, a
subquadratic algorithm for the Wiener index and the diameter of planar
graphs was presented.

\begin{figure}[t]
  \captionsetup[subfigure]{singlelinecheck=true}
  \centering
  \qquad \qquad \qquad
  \subcaptionbox{\label{fig:ExGrMed1}}
  {\includegraphics[scale=0.40]{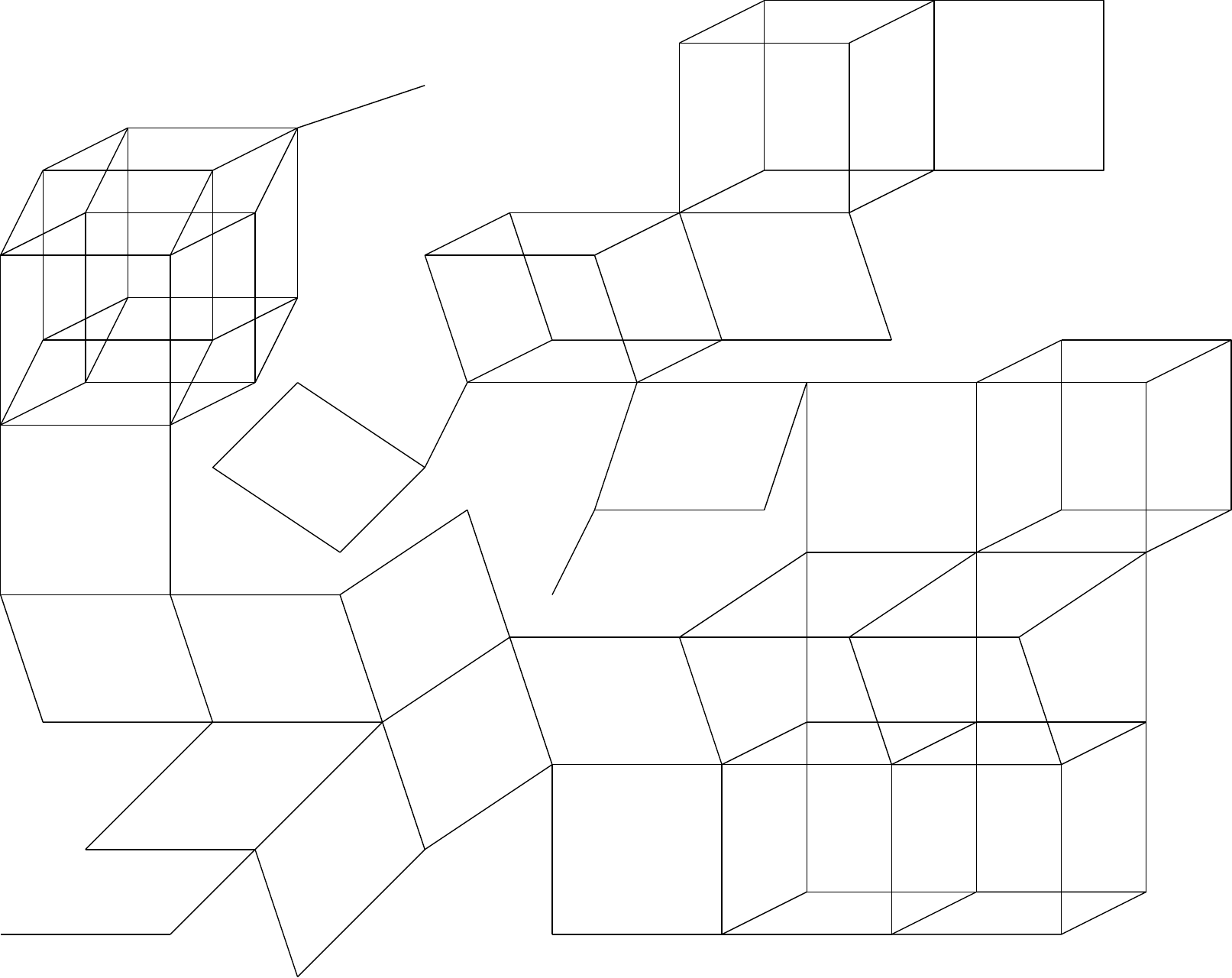}}
  \hfill \subcaptionbox{\label{fig:ExGrMed2}}
  {\includegraphics[scale=0.35,page=1]{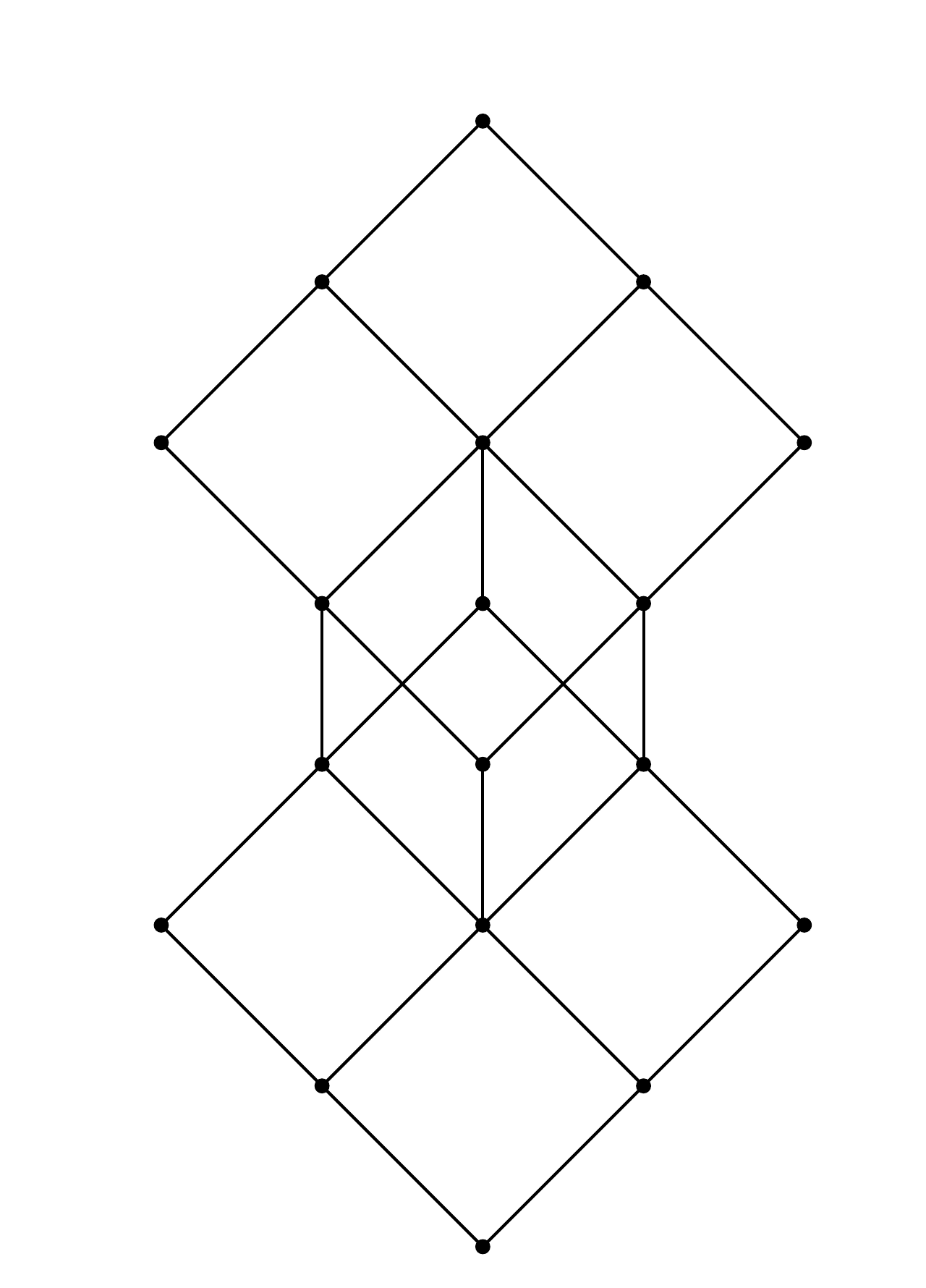}}
  \qquad \qquad \qquad
  \caption{Two median graphs, the second one (denoted by $D$) will be
    our running example.}
\end{figure}

\section{Preliminaries}
	
All graphs $G=(V,E)$ in this paper are finite, undirected, simple, and
connected; $V$ is the vertex-set and $E$ is the edge-set of $G$.  We
write $u\sim v$ if $u,v\in V$ are adjacent.  The \emph{distance}
$d(u,v)=d_G(u,v)$ between two vertices $u$ and $v$ is the length of a
shortest $(u,v)$-path, and the \emph{interval}
$I(u,v)=\{ x \in V : d(u,x) + d(x,v) = d(u,v) \}$ consists of all the
vertices on shortest $(u,v)$--paths.  A set $H$ (or the subgraph
induced by $H$) is \emph{convex} if $I(u,v)\subseteq H$ for any two
vertices $u,v$ of $H$; $H$ is a \emph{halfspace} if $H$ and
$V\setminus H$ are convex.  Finally, $H$ is \emph{gated} if for every
vertex $v \in V$, there exists a (unique) vertex $v' \in V(H)$ (the
\emph{gate} of $v$ in $H$) such that for all $u \in V(H)$,
$v' \in I(u,v)$.  The \emph{$k$-dimensional hypercube $Q_k$} has all
subsets of $\{1,\ldots,k\}$ as the vertex-set and $A \sim B$ iff
$|A \triangle B| = 1$. A graph $G$ is called \emph{median} if
$I(x,y) \cap I(y,z) \cap I(z,x)$ is a singleton for each triplet
$x,y,z$ of vertices; this unique intersection vertex $m(x,y,z)$ is
called the \emph{median} of $x,y,z$.  Median graphs are bipartite and
do not contain induced $K_{2,3}$.  The \emph{dimension} $d=\dim(G)$ of
a median graph $G$ is the largest dimension of a hypercube of $G$.  In
$G$, we refer to the $4$-cycles as \emph{squares}, and the hypercube
subgraphs as \emph{cubes}.
	
A map $w:V\rightarrow \R^+\cup \{ 0\}$ is called a \emph{weight
  function}. For a vertex $v\in V$, $w(v)$ denotes the weight of $v$
(for a set $S\subseteq V$, $w(S)=\sum_{x\in S} w(x)$ denotes the
weight of $S$).  Then $F_w(x)=\sum_{v\in V} w(v)d(x,v)$ is called the
\emph{median function} of the graph $G$ and a vertex $x$ minimizing
$F_w$ is called a \emph{median vertex} of $G$.  Finally,
$\Med_w(G)=\{x\in V : x \mbox{ is a median of } G\}$ is called the
\emph{median set} (or simply, the \emph{median}) of $G$ with respect
to the weight function $w$.  The \emph{Wiener index} $W(G)$ (called
also the \emph{total distance}) of a graph $G$ is the sum of all
pairwise distances between the vertices of $G$.  For a weight function
$w$, the \emph{Wiener index} of $G$ is the sum $W_w(G)=\sum_{u,v\in V}
w(u)w(v)d(u,v)$.

\section{Facts about median graphs}\label{sec:properties}

We recall the principal properties of median graphs used in the
algorithms.  Some of those results are a part of folklore for the
people working in metric graph theory and some other results can be
found in the papers \cite{Mu,Mu_Exp} by Mulder. For readers
convenience, we provide the proofs (sometimes different from the
original proofs) of all those results in the Appendix.
	
From now on, $G=(V,E)$ is a median graph with $n$ vertices and $m$
edges.  The first three properties follow from the definition.
	
\begin{lemma}[Quadrangle Condition]\label{quadrangle}
  For any vertices $u,v,w,z$ of $G$ such that $v, w\sim z$ and
  $d(u,v) = d(u,w) = d(u,z)-1 = k$, there is a unique vertex
  $x\sim v,w$ such that $d(u,x) = k-1$.
\end{lemma}

\begin{lemma}[Cube Condition]\label{cube}
  Any three squares of $G$, pairwise intersecting in three edges and
  all three intersecting in a single vertex, belong to a 3-dimensional
  cube of $G$.
\end{lemma}

\begin{lemma}[Convex=Gated]\label{convex-gated}
  A subgraph of $G$ is convex if and only if it is gated. 
\end{lemma}

Two edges $uv$ and $u'v'$ of $G$ are in relation $\Theta_0$ if
$uvv'u'$ is a square of $G$ and $uv$ and $u'v'$ are opposite edges of
this square. Let $\Theta$ denote the reflexive and transitive closure
of $\Theta_0$. Denote by $E_1,\ldots,E_q$ the equivalence classes of
$\Theta$ and call them \emph{$\Theta$-classes} (see
Fig.~\ref{fig-halfspaces}(a)).
	
\begin{lemma}[\!\!\cite{Mu})(Halfspaces and
  $\Theta$-classes]\label{halfspaces} For any $\Theta$-class $E_i$ of
  $G$, the graph $G_i=(V,E\setminus E_i)$ consists of exactly two
  connected components $H'_i$ and $H''_i$ that are halfspaces of $G$;
  all halfspaces of $G$ have this form.  If $uv\in E_i$, then $H'_i$
  and $H''_i$ are the subgraphs of $G$ induced by
  $W(u,v)=\{ x\in V: d(u,x)<d(v,x)\}$ and
  $W(v,u)=\{ x\in V: d(v,x)<d(u,x)\}$.
\end{lemma}

By~\cite{Dj}, $G$ is a partial cube, (i.e., isometrically embeds into
an hypercube) iff $G$ is bipartite and $W(u,v)$ is convex for any edge
$uv$ of $G$. Consequently, we obtain the following corollary.

\begin{corollary}\label{cor-halfspaces-convex}
  $G$ isometrically embeds into a hypercube of dimension equals to the
  number $q$ of $\Theta$-classes of $G$.
\end{corollary}

\begin{lemma}\label{convex-int-halfspaces}
  Each convex subgraph $S$ of $G$ is the intersection of all
  halfspaces containing $S$.
\end{lemma}

Two $\Theta$-classes $E_i$ and $E_j$ are \emph{crossing} if each
halfspace of the pair $\{ H'_i,H''_i\}$ intersects each halfspace of
the pair $\{ H'_j,H''_j\}$; otherwise, $E_i$ and $E_j$ are called
\emph{laminar}.
	
\begin{lemma}[Crossing $\Theta$-classes]\label{crossing}
  Any vertex $v \in V(G)$ and incident edges
  $vv_1\in E_1, \ldots, vv_k \in E_k$ belong to a single cube of $G$
  if and only if $E_1, \ldots, E_k$ are pairwise crossing.
\end{lemma}
	
The \emph{boundary} $\partial H'_i$ of a halfspace $H'_i$ is the
subgraph of $H'_i$ induced by all vertices $v'$ of $H'_i$ having a
neighbor $v''$ in $H''_i$.  A halfspace $H'_i$ of $G$ is
\emph{peripheral} if $\partial H'_i=H'_i$ (See
Fig.~\ref{fig-halfspaces}(b)).
	
\begin{lemma}[Boundaries]\label{boundary}
  For any $\Theta$-class $E_i$ of $G$, $\partial H'_i$ and
  $\partial H''_i$ are isomorphic and gated.
\end{lemma}
	
From now on, we suppose that $G$ is rooted at an arbitrary vertex
$v_0$ called the \emph{basepoint}. For any $\Theta$-class $E_i$, we
assume that $v_0$ belongs to the halfspace $H''_i$.  Let
$d(v_0,H'_i)=\min \{ d(v_0,x): x\in H'_i\}$.  Since $H'_i$ is gated,
the gate of $v_0$ in $H'_i$ is the unique vertex of $H'_i$ at distance
$d(v_0,H'_i)$ from $v_0$. Since median graphs are bipartite, the
choice of $v_0$ defines a canonical orientation of the edges of $G$:
$uv\in E$ is oriented from $u$ to $v$ (notation $\overrightarrow{uv}$)
if $d(v_0,u)<d(v_0,v)$. Let $\overrightarrow{G}_{v_0}$ denote the
resulting oriented pointed graph.
	
\begin{lemma}[\!\!\cite{Mu_Exp})(Peripheral Halfspaces]\label{peripheral}
  Any halfspace $H'_i$ maximizing $d(v_0,H'_i)$ is peripheral.
\end{lemma}
	
\begin{figure}[t]
  \captionsetup[subfigure]{singlelinecheck=true}
  \centering
  \subcaptionbox{\label{fig:ExGrMedGate}}{\includegraphics[scale=0.64,page=7]{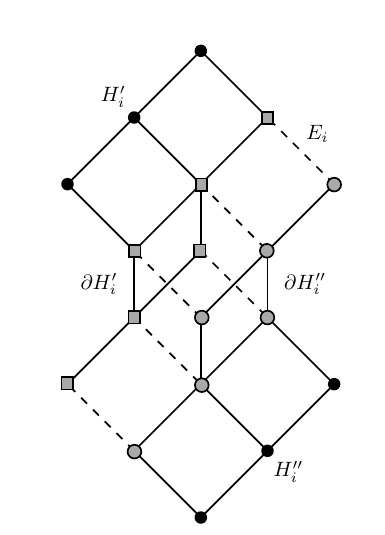}}\qquad \subcaptionbox{\label{fig:ExGrMedBoun}}{\includegraphics[scale=0.64,page=6]{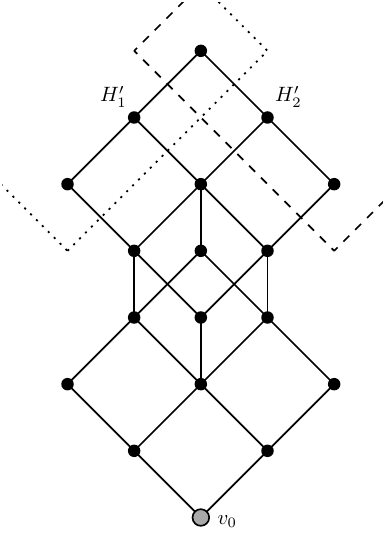}}\qquad \subcaptionbox{\label{fig:ExGrMedPer}}{\includegraphics[scale=0.64,page=5]{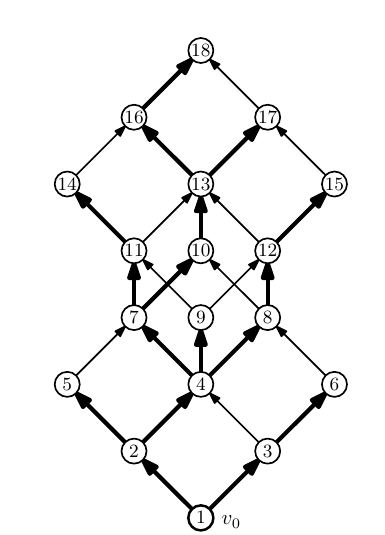}}
  \caption{~(a) In dashed, the $\Theta$-class $E_i$ of $D$, its two
    complementary halfspaces $H'_i$ and $H''_i$ and their boundaries
    $\partial H_i'$ and $\partial H_i''$,~(b) two peripheral
    halfspaces of $D$, and~(c) a LexBFS ordering of
    $D$.}\label{fig-halfspaces}
\end{figure}
	
For a vertex $v$, all vertices $u$ such that $\overrightarrow{uv}$ is
an edge of $\overrightarrow{G}_{v_0}$ are called \emph{predecessors}
of $v$ and are denoted by $\Lambda(v)$. Equivalently, $\Lambda(v)$
consists of all neighbors of $v$ in the interval $I(v_0,v)$. A median
graph $G$ satisfies the \emph{downward cube property} if any vertex
$v$ and all its predecessors $\Lambda(v)$ belong to a single cube of
$G$.
	
\begin{lemma}[\!\!\cite{Mu})(Downward Cube Property]\label{descendent_cube}
  $G$ satisfies the downward cube property.
\end{lemma}
	
Lemma~\ref{descendent_cube} immediately implies the following upper
bound on the number of edges of $G$:
	
\begin{corollary}\label{upper_edges}
  If $G$ has dimension $d$, then $m\le dn\le n\log n$.
\end{corollary}
	
We give a sharp lower bound on the number $q$ of $\Theta$-classes,
which is new to our knowledge.
	
\begin{proposition}\label{prop-nbthetaclasses}
  If $G$ has $q$ $\Theta$-classes and dimension $d$, then
  $q \geq d(\sqrt[d]{n}-1)$. This lower bound is realized for products
  of $d$ paths of length $\sqrt[d]{n}-1$.
\end{proposition}
	
\begin{proof}
  Let $\Gamma(G)$ be the crossing graph $\Gamma(G)$ of $G$:
  $V(\Gamma(G))$ is the set of $\Theta$-classes of $G$ and two
  $\Theta$-classes are adjacent in $\Gamma(G)$ if they are crossing.
  Note that $|V(\Gamma(G))| = q$.  Let $X(\Gamma(G))$ be the clique
  complex of $\Gamma(G)$. By the characterization of median graphs
  among ample classes~\cite[Proposition~4]{BaChDrKo}, the number of
  vertices of $G$ is equal to the number $|X(\Gamma(G))|$ of simplices
  of $X(\Gamma(G))$. Since $G$ is of dimension $d$,
  by~\cite[Proposition~4]{BaChDrKo}, $\Gamma(G)$ does not contain
  cliques of size $d+1$.  By Zykov's theorem~\cite{Zyk} (see
  also~\cite{Wood}), the number of $k$-simplices in $X(\Gamma(G))$ is
  at most $\binom{d}{k}\left(\frac{q}{d}\right)^k$. Hence
  $ n = |V(G)| = |X(\Gamma(G))| \leq \sum_{k=0}^d
  \binom{d}{k}\left(\frac{q}{d}\right)^k = \left(1+
    \frac{q}{d}\right)^d$ and thus $q \geq d(\sqrt[d]{n}-1)$. Let now
  $G$ be the Cartesian product of $d$ paths of length
  $(\sqrt[d]{n}-1)$. Then $G$ has $(\sqrt[d]{n}-1+1)^d = n$ vertices
  and $d(\sqrt[d]{n}-1)$ $\Theta$-classes (each $\Theta$-class of $G$
  corresponds to an edge of one of factors).
\end{proof}

\section{Computation of the $\Theta$-classes}\label{sec:theta}

In this section we describe two algorithms to compute the
$\Theta$-classes of a median graph $G$: one runs in time $O(dm)$ and
uses BFS, and the other runs in time $O(m)$ and uses LexBFS.
	
\subsection{$\Theta$-classes via BFS}
The \emph{Breadth-First Search (BFS)} refines the basepoint order and defines the same orientation
$\overrightarrow{G}_{v_0}$ of $G$.
BFS uses a queue $Q$ and the insertion in $Q$ defines a total order $<$ on the vertices of $G$: $x<y$ iff $x$ is
inserted in $Q$ before $y$. When a vertex $u$ arrives at the head of
$Q$, it is removed from $Q$ and all not yet discovered neighbors $v$
of $u$ are inserted in $Q$; $u$ becomes the \emph{parent} $f(v)$ of
$v$; for any vertex $v\neq v_0$, $f(v)$ is the smallest predecessor of
$v$. The arcs $\overrightarrow{f(v)v}$ define the \emph{BFS-tree} of
$G$.
For each vertex $v$, BFS produces the list $\Lambda(v)$ of
predecessors of $v$ ordered by $<$; denote this ordered list by
$\Lambda_{<}(v)$.
By Lemma~\ref{descendent_cube}, each list $\Lambda_{<}(v)$ has size at
most $d:=\dim(G)$.  Notice also that the total order $<$ on vertices
of $G$ give raise to a total order on the edges of $G$: for two edges
$uv$ and $u'v'$ with $u<v$ and $u'<v'$ we have $uv<u'v'$ if and only
if $u<u'$ or if $u=u'$ and $v<v'$.
	
Now we show how to use a BFS rooted at $v_0$ to compute, for each edge
$uv$ of a median graph $G$, the unique $\Theta$-class $E(uv)$
containing the edge $uv$. Suppose that $uv$ is oriented by BFS from
$u$ to $v$, i.e., $d(v_0,u)<d(v_0,v)$. There are only two
possibilities: either the edge $uv$ is the first edge of the
$\Theta$-class $E(uv)$ discovered by BFS or the $\Theta$-class of $uv$
already exists. The following lemma shows how to distinguish between
these two cases:

\begin{lemma}\label{prime_traces}
  An edge $uv\in E_i$ with $d(v_0,u)<d(v_0,v)$ is the first edge of a
  $\Theta$-class $E_i$ iff $u$ is the unique predecessor of $v$, i.e.,
  $\Lambda_<(v)=\{ u\}$.
\end{lemma}

\begin{proof}
  First let $uv$ be the first edge of $E_i$ discovered by {BFS}. Since
  $H'_i$ is gated, $v$ is the gate of $v_0$ in $H'_i$ and $u$ is the
  unique neighbor of $v$ in $H''_i$. We assert that $u$ is the unique
  neighbor of $v$ in $I(v_0,v)$. Suppose $I(v_0,v)$ contains a second
  neighbor $u''$ of $v$. Since $v$ is the gate of $v_0$ in $H'_i$ and
  $u''$ is closer to $v_0$ than $v$, $u''$ necessarily belongs to
  $H''_i$, a contradiction with the uniqueness of $u$.  Conversely,
  suppose that $v$ has only $u$ as a neighbor in $I(v_0,v)$ but $uv$
  is not the first edge of $E_i$ with respect to the BFS. This implies
  that the gate $x$ of $v_0$ in $H'_i$ is different from $v$.  Let
  $u'$ be a neighbor of $v$ in $I(x,v)$ and note that
  $I(x,v)\subseteq I(v_0,v)$. Since $v,x\in H'_i$ and $H'_i$ is
  convex, $u'$ belongs to $H'_i$. Since $u$ belongs to $H''_i$, we
  conclude that $u$ and $u'$ are two different neighbors of $v$ in
  $I(v_0,v)$, a contradiction.
\end{proof}
	
If $uv$ is not the first edge of its $\Theta$-class, the following
lemma shows how to find its $\Theta$-class:

\begin{lemma}\label{nonprime_traces}
  Let $uv$ be an edge of a median graph with $u\in \Lambda_<(v)$. If
  $v$ has a second predecessor $v'$, then there exists a square
  $u'uvv'$ in which $uv$ and $u'v'$ are opposite edges and
  $u'\in \Lambda_<(u)\cap \Lambda_< (v')$.
\end{lemma}
	
\begin{proof}
  Indeed, by the quadrangle condition, the vertices $u$ and $v'$ have
  a unique common neighbor $u'$ such that $u'uvv'$ is a square of $G$
  and $u'$ is closer to $v_0$ than $u$ and $v'$. Consequently,
  $u'\in \Lambda_<(u)\cap \Lambda_<(v')$ and $uv$ and $u'v'$ are
  opposite edges of $u'uvv'$.
\end{proof}
	
From Lemmas~\ref{prime_traces} and~\ref{nonprime_traces} we deduce the
following algorithm for computing the $\Theta$-classes of $G$. First,
run a BFS and return a BFS-ordering of the vertices and edges of $G$
and the ordered lists $\Lambda_<(v), v\in V$. Then consider the edges
of $G$ in the BFS-order. Pick a current edge $uv$ and suppose that
$u\in \Lambda_<(v)$. If $\Lambda_<(v)=\{u\}$, by
Lemma~\ref{prime_traces} $uv$ is the first edge of its $\Theta$-class,
thus create a new $\Theta$-class $E_i$ and insert $uv$ in
$E_i$. Otherwise, if $v$ has a second predecessor $v'$, then traverse
the ordered lists $\Lambda_<(u)$ and $\Lambda_<(v')$ to find their
unique common predecessor $u'$ (which exists by
Lemma~\ref{nonprime_traces}). Then insert the edge $uv$ in the
$\Theta$-class of the edge $u'v'$. Since the two sorted lists
$\Lambda_<(u)$ and $\Lambda_<(v')$ are of size at most $d$, their
intersection (that contains only $u'$) can be computed in time $O(d)$,
and thus the $\Theta$-class of each edge $uv$ of $G$ can be computed
in $O(d)$ time. Consequently, we obtain:
	
\begin{proposition}\label{BFS}
  The $\Theta$-classes of a median graph $G$ with $n$ vertices, $m$
  edges, and dimension $d$ can be computed in $O(dm)=O(d^2n)$ time.
\end{proposition}
	
\subsection{$\Theta$-classes via LexBFS}
The \emph{Lexicographic Breadth-First Search (LexBFS)}, proposed
in~\cite{RoTaLu}, is a refinement of BFS.  In BFS, if $u$ and $v$ have
the same parent, then the algorithm order them arbitrarily. Instead,
the LexBFS chooses between $u$ and $v$ by considering the ordering of
their second-earliest predecessors. If only one of them has a
second-earliest predecessor, then that one is chosen. If both $u$ and
$v$ have the same second-earliest predecessor, then the tie is broken
by considering their third-earliest predecessor, and so on (See
Fig.~\ref{fig-halfspaces}(c)).  The LexBFS uses a set partitioning
data structure and can be implemented in linear time~\cite{RoTaLu}.
In median graphs, the next lemma shows that it suffices to consider
only the earliest and second-earliest predecessors, leading to a
simpler implementation of LexBFS:
	
\begin{lemma}\label{LexBFSmedian}
  If $u$ and $v$ are two vertices of a median graph $G$, then
  $|\Lambda (u)\cap \Lambda (v)|\le 1$.
\end{lemma}

\begin{proof}
  Let $x,x'$ be two distinct predecessors of $u$ and $v$. Since
  $x,x'\in \Lambda (u)\cap \Lambda (v)$, we have
  $d(v_0,u)=d(v_0,v)=d(v_0,x)+1=d(v_0,x')+1=k+1$. By
  Lemma~\ref{quadrangle}, there is a vertex $y\sim x,x'$ at distance
  $k-1$ from $v_0$. But then $x,x',u,v,y$ induce a forbidden
  $K_{2,3}$.
\end{proof}

A graph $G$ satisfies the \emph{fellow-traveler property} if for any
LexBFS ordering of the vertices of $G$, for any edge $uv$ with
$v_0 \notin\{u,v\}$, the parents $f(u)$ and $f(v)$ are adjacent.
	
\begin{theorem}\label{fellow-traveler}
  Any median graph $G$ satisfies the fellow-traveler property.
\end{theorem}
	
\begin{proof}
  Let $<$ be an arbitrary LexBFS order of the vertices of $G$ and $f$
  be its parent map.  Since any LexBFS order is a BFS order, $<$ and
  $f$ satisfy the following properties of BFS:
\begin{enumerate}[({BFS}1)]
  \item if $u<v$, then $f(u) \leq f(v)$;
  \item if $f(u)<f(v)$, then $u<v$;
  \item if $v\neq v_0$, then
    $f(v)=\min_<\{u: u\sim v\}$;
  \item if $u<v$ and $v \sim f(u)$, then $f(v)=f(u)$.
  \end{enumerate}

  Notice also the following simple but useful property:
		
  \begin{lemma}\label{claim1}
    If $abcd$ is a square of $G$ with $d(v_0,c)=k$,
    $d(v_0,b)=d(v_0,d)=k+1, d(v_0,a)=k+2$ and $f(a)=b$, and the edge
    $ad$ satisfies the fellow-traveler property, then $f(d)=c$.
  \end{lemma}

  \begin{proof}
    By the fellow traveler property, $f(d) \sim f(a) = b$. If
    $f(d) \neq c$, then $a, b, c, d, f(d)$ induce a forbidden
    $K_{2,3}$.
  \end{proof}

  \begin{figure}
    \captionsetup[subfigure]{singlelinecheck=true}
    \centering
    \subcaptionbox{\label{fig:lexcara}}
    {\includegraphics[scale=0.6,page=1]{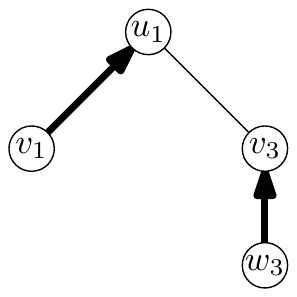}}
    \hfill \subcaptionbox{\label{fig:lexcarb}}
    {\includegraphics[scale=0.6,page=2]{Images/FT-Property.pdf}}
    \hfill \subcaptionbox{\label{fig:lexcarc}}
    {\includegraphics[scale=0.6,page=3]{Images/FT-Property.pdf}}
    \hfill \subcaptionbox{\label{fig:lexard}}
    {\includegraphics[scale=0.6,page=4]{Images/FT-Property.pdf}}
    \hfill \subcaptionbox{\label{fig:lexcare}}
    {\includegraphics[scale=0.6,page=5]{Images/FT-Property.pdf}}
    \vspace{3ex}
    
    \subcaptionbox{\label{fig:lexcarf}}
    {\includegraphics[scale=0.6,page=6]{Images/FT-Property.pdf}}
    \hfill \subcaptionbox{\label{fig:lexcarg}}
    {\includegraphics[scale=0.6,page=8]{Images/FT-Property.pdf}}
    \hfill \subcaptionbox{\label{fig:lexcarh}}
    {\includegraphics[scale=0.6,page=9]{Images/FT-Property.pdf}}
\hfill \subcaptionbox{\label{fig:lexcarj}}
    {\includegraphics[scale=0.6,page=11]{Images/FT-Property.pdf}}	
    \caption{Animated proof of
      Theorem~\ref{fellow-traveler}.}\label{FT-Property}
  \end{figure}
  
  We prove the fellow-traveler property by induction on the total
  order on the edges of $G$ defined by $<$. The proof is illustrated
  by several figures (the arcs of the parent map are represented in
  bold).  We use the following convention: all vertices having the
  same distance to the basepoint $v_0$ will be labeled by the same
  letter but will be indexed differently; for example, $w_1$ and $w_2$
  are two vertices having the same distance to $v_0$.
		
  Suppose by way of contradiction that $e=u_1v_3$ with $v_3<u_1$ is
  the first edge in the order $<$ such that the parents $f(u_1)$ and
  $f(v_3)$ of $u_1$ and $v_3$ are not adjacent. Then necessarily
  $f(u_1)\ne v_3$.  Set $v_1=f(u_1)$ and $w_3=f(v_3)$
  (Fig.~\ref{fig:lexcara}).  Since $d(v_0, v_1)=d(v_0,v_3)$ and
  $u_1\sim v_1,v_3$, by the quadrangle condition $v_1$ and $v_3$ have
  a common neighbor at distance $d(v_0,v_1)-1$ from $v_0$. This vertex
  cannot be $w_3$, otherwise $f(u_1)$ and $f(v_3)$ would be
  adjacent. Therefore there is a vertex $w_4\sim v_1,v_3$ at distance
  $d(v_0,v_1)-1$ from $v_0$ (Fig.~\ref{fig:lexcarb}). By induction
  hypothesis, the parent $x_3=f(w_4)$ of $w_4$ is adjacent to
  $w_3=f(v_3)$.  Since $u_1\sim v_1=f(u_1),v_3$ and
  $v_3\sim w_3=f(v_3),w_4$, by (BFS3) we conclude that $v_1 < v_3$ and
  $w_3<w_4$.  By (BFS2), $f(v_1)\leq f(v_3)$, whence $f(v_1)\leq w_3$
  and since $f(v_1)\neq f(v_3)$ (otherwise, $f(u_1)\sim f(v_3)$), we
  deduce that $f(v_1)< w_3<w_4$. Hence $f(v_1) \neq w_4$. Set
  $w_1=f(v_1)$. By the induction hypothesis, $f(v_1)=w_1$ is adjacent
  to $f(w_4)=x_3$ (Fig.~\ref{fig:lexcarc}). By the cube condition
  applied to the squares $w_4v_1w_1x_3$, $w_4v_1u_1v_3$, and
  $w_4v_3w_3x_3$ there is a vertex $v_2$ adjacent to $u_1$, $w_1$, and
  $w_3$.  Since $u_1\sim v_2$ and $f(u_1)=v_1$, by (BFS3) we obtain
  $v_1<v_2$. Since $v_2$ is adjacent to $w_1$ and $w_1=f(v_1)$, by
  (BFS4) we obtain $f(v_2)=f(v_1)=w_1$, and by (BFS2),
  $v_2<v_3$. Since $f(v_2)=w_1$, by Lemma~\ref{claim1} for $v_2w_1x_3w_3$, we obtain $f(w_3)=x_3$ (Fig.~\ref{fig:lexard}).
  Since $v_1<v_2$, $f(v_1)=f(v_2)=w_1$, and $v_2\sim w_1,w_3$, by
  LexBFS $v_1$ is adjacent to a predecessor different from $w_1$ and
  smaller than $w_3$. Since $w_3<w_4$, this predecessor cannot be
  $w_4$. Denote by $w_2$ the second smallest predecessor of $v_1$
  (Fig.~\ref{fig:lexcare}) and note that $w_1 < w_2 < w_3 <w_4$.
		
  By the quadrangle condition, $w_2$ and $w_4$ are adjacent to a
  vertex $x_5$, which is necessarily different from $x_3$ because $G$
  is $K_{2,3}$-free. By the induction hypothesis, $f(w_2)$ and
  $f(v_1)=w_1$ are adjacent.  Then $f(w_2)\ne x_3,x_5$, otherwise we
  obtain a forbidden $K_{2,3}$. Set $f(w_2)=x_2$. Analogously,
  $f(x_5)=y_5$ and $f(w_2)=x_2$ are adjacent as well as $f(x_5)=y_5$
  and $f(w_4)=x_3$ (Fig.~\ref{fig:lexcarf}). By (BFS1),
  $x_2 = f(w_2) < f(w_3) = x_3$ and by (BFS3), $x_3 = f(w_4) < x_5$.
  Since $w_3<w_4$ with $f(w_3)=f(w_4)$ and $w_4$ is adjacent to $x_5$,
  by LexBFS $w_3$ must have a predecessor different from $x_3$ and
  smaller than $x_5$. This vertex cannot be $x_2$ by (BFS3) since
  $f(w_3) = x_3$.  Denote this predecessor of $w_3$ by $x_4$ and
  observe that $x_2 <x_3<x_4<x_5$. By the induction hypothesis, the
  parent of $x_4$ is adjacent to $f(w_3)=x_3$. Let $y_4=f(x_4)$.
		
  If $y_4=y_5$, applying the cube condition to the squares
  $x_3w_3x_4y_5$, $x_3w_4x_5y_5$, and $x_3w_4v_3w_3$ we find a vertex
  $w$ adjacent to $x_4$, $v_3$, and $x_5$.  Applying the cube
  condition to the squares $w_4v_3wx_5$, $w_4v_1w_2x_5$, and
  $w_4v_1u_1v_3$ we find a vertex $v$ adjacent to $u_1$, $w_2$, and
  $w$.  Since $v\sim w_2$, by (BFS3) $f(v)\leq w_2<w_3=f(v_3)$, hence
  by (BFS2) we obtain $v<v_3$. Therefore we can apply the induction
  hypothesis, and by Lemma~\ref{claim1} for $u_1v_1w_2v$, we deduce that $f(v)=w_2$.  By Lemma~\ref{claim1}
for $v_3w_3x_4w$, we deduce that $f(w)=x_4$
  (Fig.~\ref{fig:lexcarg}). Applying the induction hypothesis to the
  edge $vw$ we have that $f(v)=w_2$ is adjacent to $f(w) = x_4$,
  yielding a forbidden $K_{2,3}$ induced by $v, x_5, x_4, w, w_2$
  (Fig.~\ref{fig:lexcarg}). All this shows that $y_4\neq y_5$. By the
  quadrangle condition, $y_5$ and $y_4$ have a common neighbor $z_3$
  (Fig.~\ref{fig:lexcarh}).
		
Recall that $x_2<x_3<x_4<x_5$, and note that by (BFS1),
  $y_4=f(x_4)<f(x_5)=y_5$.  We denote by $H$ the subgraph of $G$
  induced by the vertices
  $V'=\{w_1,x_2,x_3,x_4,x_5,y_4,y_5,z_3\}$. The set of edges of $H$ is
  $E'=\{z_3y_4, z_3y_5, y_4x_3, y_4x_4, y_5x_2, y_5x_3, y_5x_5,
  x_2w_1, x_3w_1\}$.  To conclude the proof, we use the following technical
  lemma.
  \begin{lemma}\label{claim2}
    Let $H=(V',E')$ (Fig.~\ref{fig-lem-aux-lexBFS-a}) be an induced
    graph of $G$, where
    $d(v_0, w_1)=d(v_0, x_2)+1=\cdots=d(v_0, x_5)+1=d(v_0,
    y_4)+2=d(v_0, y_5)+2=d(v_0, z_3)+3$ and $f(x_5)=y_5$ and
    $f(x_4)=y_4$, such that $x_2<x_3<x_4<x_5$ and $y_4<y_5$. If $G$
    satisfies the fellow-traveler property up to distance
    $d(v_0, w_1)$, then there exists a vertex $x_0$ such that
    $x_0<x_2$ and $x_0\sim w_1,y_4$ (Fig.~\ref{fig-lem-aux-lexBFS-b}).
  \end{lemma}

  \begin{figure}[h]
    \captionsetup[subfigure]{singlelinecheck=true}
    \centering
    \subcaptionbox{\label{fig-preuve-aux-lexBFS-c}}
    {\includegraphics[scale=0.6,page=16]{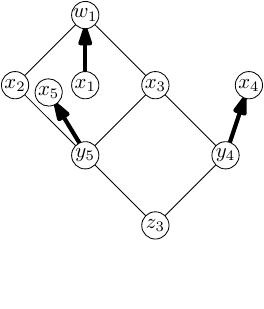}}\hfill
    \subcaptionbox{\label{fig-preuve-aux-lexBFS-d}}
    {\includegraphics[scale=0.6,page=17]{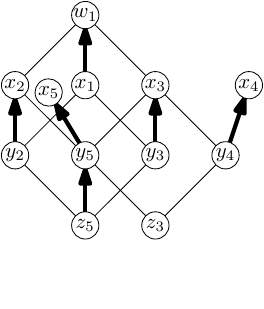}}\hfill
    \subcaptionbox{\label{fig-preuve-aux-lexBFS-e}}
    {\includegraphics[scale=0.6,page=18]{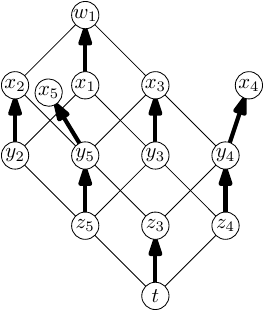}}\hfill
    \subcaptionbox{\label{fig-preuve-aux-lexBFS-g}}
    {\includegraphics[scale=0.6,page=20]{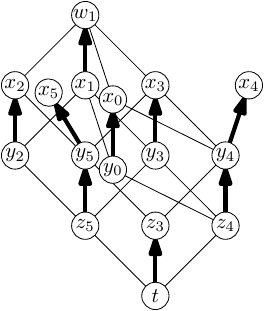}}
    \caption{To the proof of Lemma~\ref{claim2}.}\label{fig-preuve-aux-lexBFS}
  \end{figure}

  \begin{proof}[Proof of Lemma~\ref{claim2}]
Consider a median graph $G$ for which Lemma~\ref{claim2} does not
    hold. Among all induced subgraphs of $G$ satisfying the conditions
    of the lemma but for which there does not exist a vertex
    $x_0\ne x_3\sim w_1,y_4$ with $x_0<x_2$, we select a copy of $H$
    minimizing the distance $d(v_0,w_1)$. First, suppose that
    $f(w_1)=x_2$. By Lemma~\ref{claim1} for $w_1x_2y_5x_3$, we deduce $f(x_3)=y_5$. Then, by (BFS1), we
    get $y_5=f(x_3)\leq f(x_4) \leq f(x_5)=y_5$. Hence, $f(x_4)=y_5$,
    a contradiction. Therefore $f(w_1) \neq x_2$. Since $G$ satisfies
    the fellow-traveler property up to distance $d(v_0,w_1)$, we get
    $f(x_2)\sim f(w_1)$. Let $x_1$ be the parent of $w_1$
    (Fig.~\ref{fig-preuve-aux-lexBFS-c}) and let $y_2 = f(x_2)$ be the
    parent of $x_2$. To avoid an induced $K_{2,3}$, $y_2$ cannot
    coincide with $y_5$. Moreover, $y_2$ does not coincide with $y_4$
    because otherwise $x_1$ would be the common neighbor of $w_1$ and
    $y_4$ required by Lemma~\ref{claim2}.  Let $z_5$ be the parent of
    $y_5$. By the fellow-traveler property, $z_5 = f(y_5)$ is adjacent
    to $y_2=f(x_2)$.  By the cube condition applied to the squares
    $x_2w_1x_1y_2$, $x_2w_1x_3y_5$, and $x_2y_2z_5y_5$, we find a
    neighbor $y_3$ of $x_3$, $x_1$, and $z_5$. If $z_5 = z_3$, then
    $y_3=y_4$ (otherwise we get a $K_{2,3}$) and $x_1$ is the neighbor
    of $w_1$ and $y_4$ required by Lemma~\ref{claim2}, a
    contradiction. Thus $y_3 \neq y_4$ and $z_5 \neq z_3$. Moreover,
    by Lemma~\ref{claim1} for $w_1x_1y_3x_3$, $y_3=f(x_3)$ (see
    Fig~\ref{fig-preuve-aux-lexBFS-d}). Let $t$ be the parent of
    $z_3$. By induction hypothesis, $z_5=f(y_5)\sim
    t=f(z_3)$. Applying the cube condition to the squares
    $y_5z_3tz_5$, $y_5x_3y_3z_5$, and $y_5x_3y_4z_3$, we find a
    neighbor $z_4$ of $t$, $y_3$ and $y_4$. By Lemma~\ref{claim1}
    for $x_3y_3z_4y_4$, $f(y_4)=z_4$ (Fig.~\ref{fig-preuve-aux-lexBFS-e})
    and by (BFS1), $x_2<x_3<x_4<x_5$ implies
    $y_2=f(x_2)< y_3=f(x_3)<y_4=f(x_4)< y_5=f(x_5)$. Since
    $d(x_1, v_0)<d(w_1,v_0)$, our choice of $H$ implies the existence
    of a neighbor $y_0$ of $x_1$ and $z_4$ such that $y_0<y_2$
    (Fig.~\ref{fig-preuve-aux-lexBFS-g}). Applying the cube condition
    to the squares $y_3x_1y_0z_4$, $y_3x_1w_1x_3$ and $y_3x_3y_4z_4$,
    we find a neighbor $x_0$ of $w_1$, $y_4$, and $y_0$. By (BFS3),
    $f(x_0)\leq y_0 < y_2 = f(x_2)$ and thus, by (BFS2), $x_0<x_2$
    (Fig.~\ref{fig-preuve-aux-lexBFS-g}), a contradiction with the
    choice of $H$.
  \end{proof}
		
  Since $G$ contains a subgraph $H$ satisfying the conditions of
  Lemma~\ref{claim2}, there exists a vertex $x_0$ such that $x_0<x_2$
  and $x_0\sim w_1,y_4$ (Fig.~\ref{fig:lexcarj}).  By the cube
  condition applied to the squares $x_3w_1x_0y_4$, $x_3w_1v_2w_3$, and
  $x_3w_3x_4y_4$, there exists $w_0\sim x_0,v_2,x_4$
  (Fig.~\ref{fig:lexcarj}). Since $x_0$ is adjacent to $w_0$, by
  (BFS3) $f(w_0)\leq x_0<x_2=f(w_2)$. By (BFS2), $w_0<w_2$. Recall
  that $f(v_1)=w_1=f(v_2)$ and that $w_2$ is the second-earliest
  predecessor of $v_1$. Since $w_0<w_2$ and $w_0$ is a predecessor of
  $v_2$, by LexBFS we deduce that $v_2<v_1$. Since $v_1$ and $v_2$ are
  both adjacent to $u_1$ we obtain a contradiction with
  $f(u_1)=v_1$. This contradiction shows that any median graph $G$
  satisfies the fellow-traveller property. This finishes the proof of
  Theorem~\ref{fellow-traveler}.
\end{proof}

\begin{figure}
  \captionsetup[subfigure]{singlelinecheck=true}
  \centering
  \subcaptionbox{\label{fig-lem-aux-lexBFS-a}}
  {\includegraphics[scale=0.6,page=12]{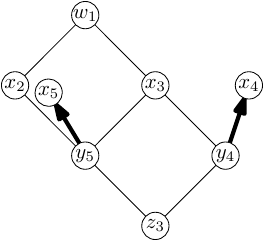}}
  \qquad\qquad  \subcaptionbox{\label{fig-lem-aux-lexBFS-b}}
  {\includegraphics[scale=0.6,page=13]{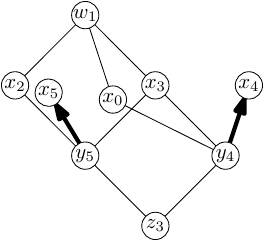}}
  \caption{The induced subgraph $H$ in Lemma~\ref{claim2}.}\label{fig-lem-aux-lexBFS}
\end{figure}
	
We now explain how to implement LexBFS in a median graph $G$ in a
simpler way than in the general case.  By Lemma~\ref{LexBFSmedian}, it
suffices to keep for each vertex $v$ only its earliest and
second-earliest predecessors, i.e., if $v$ and $w$ have the same earliest predecessor, then LexBFS will
order $v$ before $w$ iff either the second-earliest predecessor of $v$
is ordered before the second earliest predecessor of $w$ or if $v$ has
a second-earliest predecessor and $w$ does not. Similarly to BFS,
LexBFS can be implemented using a single queue $Q$. Additionally to
BFS, each already labeled vertex $u$ must store the position $\pi(u)$
in $Q$ of the earliest vertex of $Q$ having $u$ as a single
predecessor.  In $Q$, all vertices having $u$ as their parent occur
consecutively. Additionally, among these vertices, the ones having a
second predecessor must occur before the vertices having only $u$ as a
predecessor and the vertices having a second predecessor must be
ordered according to that second predecessor. To ensure this property,
we use the following rule: if a vertex $v$ in $Q$, currently having
only $u$ as a predecessor, discovers yet another predecessor $u'$,
then $v$ is swapped in $Q$ with the vertex $\pi(u)$, and $\pi(u)$ is
updated. Clearly this is an $O(m)$ implementation.

\begin{algorithm}[ht]\caption{$\Theta$-classes via LexBFS}\label{alg:calcthlin}
  \DontPrintSemicolon
  \SetAlgoVlined
  \KwData{$G=(V,E)$, $v_0\in V$}
  \KwResult{The $\Theta$-classes $\Theta$ of $G$ ordered by
    increasing distance from $v_0$}
                      
  \Begin{
    $\Theta \leftarrow \emptyset$\;
    $(E, \Lambda, f) \leftarrow LexBFS (G,v_0)$\;
    \tcp{$E$ : the list of edges ordered by LexBFS}
    \tcp{$\Lambda:V\mapsto 2^V$ such that $\Lambda(v)$ is the set of
      predecessors of $v$}
    \tcp{$f:V\mapsto V$ such that $f(v)$ is the parent of $v$}
    \ForEach{$uv \in E$ }
    {
      \uIf {$|\Lambda[v]|=1$}
      {
        Add a new $\Theta$-class $\{uv\}$ to $\Theta$
        \tcp*{first edge in the $\Theta$-class}
      }
      \uElseIf{$f(v)\neq u$}{
        Add the edge $uv$ to the $\Theta$-class of the edge
        $f(u)f(v)$
      }
      \Else{
        Pick any $x$ in $\Lambda(v)\setminus\{u\}$\;
        Add the edge $uv$ to the $\Theta$-class of the edge
        $f(x)x$\;
      }
    }
    \Return{$\Theta$}
  }
\end{algorithm}

Now we use Theorem~\ref{fellow-traveler} to compute the
$\Theta$-classes of $G$. We run LexBFS and return a LexBFS-ordering of $V(G)$ and
$E(G)$ and the ordered lists $\Lambda_<(v), v\in V$. Then consider the edges
of $G$ in the LexBFS-order. Pick the first unprocessed edge $uv$ and
suppose that $u\in \Lambda_<(v)$.  If $\Lambda_<(v)=\{u\}$, by
Lemma~\ref{prime_traces}, $uv$ is the first edge of its
$\Theta$-class, thus we create a new $\Theta$-class $E_i$ and insert
$uv$ as the first edge of $E_i$. We call $uv$ the \emph{root} of $E_i$
and keep $d(v_0,v)$ as the distance from $v_0$ to $H'_i$. Now suppose
$|\Lambda_<(v)|\ge 2$.  We consider two cases: (i) $u\neq f(v)$ and
(ii) $u=f(v)$. For (i), by
Theorem~\ref{fellow-traveler},
$uv$ and $f(u)f(v)$ are opposite edges of a square. Therefore $uv$
belongs to the $\Theta$-class of $f(u)f(v)$ (which was already
computed because $f(u)f(v)<uv$). In order to recover the $\Theta$-class of the edge $f(u)f(v)$ in
constant time, we use a (non-initialized) matrix $A$ whose rows and
columns correspond to the vertices of $G$ such that $A[x,y]$ contains
the $\Theta$-class of the edge $xy$ when $x$ and $y$ are adjacent and
the $\Theta$-class of $xy$ has already been computed and $A[x,y]$ is
undefined if $x$ and $y$ are not adjacent or if the $\Theta$-class of
$xy$ has not been computed yet. For (ii), pick any
$x\in \Lambda_<(v), x\ne u$. By Theorem~\ref{fellow-traveler},
$uv=f(v)v$ and $f(x)x$ are opposite edges of a square. Since $f(x)x$
appears before $uv$ in the LexBFS order, the $\Theta$-class of $f(x)x$
has already been computed, and the algorithm inserts $uv$ in the
$\Theta$-class of $f(x)x$. Each $\Theta$-class $E_i$ is totally ordered by the order in which the
edges are inserted in $E_i$.  Consequently, we obtain:

\begin{theorem}\label{LexBFS}
  The $\Theta$-classes of a median graph $G$ can be computed in $O(m)$ time.
\end{theorem}
	
\section{The median of $G$}\label{sec:medianWiener}
	
We use Theorem~\ref{LexBFS} to compute the median set
$\Med_w(G)$ of a median graph $G$ in $O(m)$ time.
We also use the existence of peripheral	halfspaces
and the majority rule. 
\subsection{Peripheral peeling} The order $E_1,E_2,\ldots, E_q$ in which the $\Theta$-classes $E_i$ of
$G$ are constructed corresponds to the order of the distances from
$v_0$ to $H'_i$: if $i<j$ then $d(v_0,H'_i)\le d(v_0,H'_j)$ (recall
that $v_0\in H''_i$). By Lemma~\ref{peripheral}, the halfspace $H'_q$
of $E_q$ is peripheral.  If we contract all edges of $E_q$ (i.e., we
identify the vertices of $H'_q=\partial H'_q$ with their neighbors in
$\partial H''_q$) we get a smaller median graph
$\tG=H''_q$; $\tG$ has $q-1$ $\Theta$-classes $\tE_1,\ldots,\tE_{q-1}$, where
$\tE_i$ consists of the edges of $E_i$ in $\tG$. The halfspaces of $\tG$ have the form $\tH'_i=H'_i\cap H''_q$
and $\tH''_i=H''_i\cap H''_q$. Then $\tE_1,\ldots,\tE_{q-1}$ corresponds to the ordering of the
halfspaces $\tH'_1,\ldots,\tH'_{q-1}$ of $\tG$ by their distances to
$v_0$. Hence the last halfspace $\tH'_{q-1}$ is peripheral in $\tG$. Thus the ordering $E_q,E_{q-1},\ldots,E_1$ of the $\Theta$-classes of $G$
provides us with a set $G_q=G,G_{q-1}=\tG,\ldots,G_0$ of median graphs
such that $G_0$ is a single vertex and for each $i\ge 1$, the
$\Theta$-class $E_i$ defines a peripheral halfspace in the graph $G_i$
obtained after the successive contractions of the peripheral
halfspaces of $G_q,G_{q-1},\ldots, G_{i+1}$ defined by
$E_q, E_{q-1}, \ldots, E_{i+1}$.  We call $G_q,G_{q-1},\ldots,G_0$ a
\emph{peripheral peeling} of $G$.
Since each vertex of $G$ and each $\Theta$-class is contracted only
once, we do not need to explicitly compute the restriction of each
$\Theta$-class of $G$ to each $G_i$. For this it is enough to keep for each vertex $v$ a variable
indicating whether this vertex belongs to an already contracted
peripheral halfspace or not.
Hence, when the $i$th $\Theta$-class must be contracted, we simply
traverse the edges of $E_i$ and select those edges whose both ends are
not yet contracted.

\subsection{Computing the weights of the halfspaces of
  $G$}\label{ssec:weight-halfpaces}
	
We use a peripheral peeling $G_q,G_{q-1},\ldots,G_0$ of $G$ to compute
the weights $w(H'_i)$ and $w(H''_i)$, $i=1,\ldots,q$ of all halfspaces
of $G$. As above, let $\tG$ be obtained from $G$ by contracting the
$\Theta$-class $E_q$. Consider the weight function $\tw$ on
$\tG=H''_q$ defined as follows:
\begin{equation}\label{weightsecond}
  \tw(v'') =
  \begin{cases}
    w(v'') + w(v') & \text{if } v''\in \partial H''_q, v'\in H'_q, \text{ and }
    v''\sim v',\\
    w(v'') & \text{if } v''\in H''_q \setminus \partial H''_q.\\
  \end{cases}
\end{equation}
	
\begin{algorithm}[]\caption{ComputeWeightsOfHalfspaces($G,w,\Theta$)}\label{alg:compweight} 
\SetKw{Add}{add}
  \SetKw{To}{to}
  \DontPrintSemicolon
  \SetAlgoVlined
  \KwData{A median graph $G=(V,E)$, a weight function $w : V\rightarrow
    \R^+n\cup \{0\}$, the $\Theta$-classes $\Theta = (E_1, \hdots,E_q)$ of
    $G$ ordered by increasing distance to the basepoint $v_0$.}
  \KwResult{The list of the pairs of weights
    $\left((w(H''_q),w(H'_q)), \ldots, (w(H''_1),w(H'_1))\right)$}
  \Begin{
\uIf{$|V| = 1$}{
      \Return the empty list
    }
    \Else
    {			
      Let $H'$ and $H''$ be the two complementary halfspaces defined
      byn$E_q$ ($v_0\in H''$)\;		
      $w(H') \leftarrow \sum_{v\in H'} w(v)$\;
      $w(H'') \leftarrow w(V) - w(H')$\;
\ForEach{$v'v''\in E_q$ with $v' \in H'$ and $v'' \in H''$}
      {
$w(v'')\leftarrow w(v')+w(v'')$ \;
      }
$L \leftarrow$ ComputeWeightsOfHalfspaces($H'', w,
      \Theta\setminus \{E_q\}$) \;
      \Add $(w(H''),w(H'))$ \To $L$\;
      \Return $L$\;
    }
  }
\end{algorithm}

\begin{lemma}\label{weight-halfspaces}
  For any $\Theta$-class $\tE_i$ of $\tG$, $\tw(\tH'_i)=w(H'_i)$ and
  $\tw(H''_i)=w(H''_i)$.
\end{lemma}

By Lemma~\ref{weight-halfspaces}, to compute all $w(H'_i)$ and
$w(H''_i),$ it suffices to compute the weight of the peripheral
halfspace of $E_i$ in the graph $G_i$, set it as $w(H'_i)$, and set
$w(H''_i):=w(G)-w(H'_i)$.

Let $G$ be the current median graph, let $H'_q$ be a peripheral
halfspace of $G$, and $\tG=H''_q$ be the graph obtained from $G$ by
contracting the edges of $E_q$. To compute $w(H'_q)$, we traverse the
vertices of $H'_q$ (by considering the edges of $E_q$).  Set
$w(H''_q)=w(G)-w(H'_q)$. Let $\tw$ be the weight function on $\tG$
defined by Equation~\ref{weightsecond}. Clearly, $\tw$ can be computed
in $O(|V(H'_q)|)=O(|E_q|)$ time. Then by Lemma~\ref{weight-halfspaces}
it suffices to recursively apply the algorithm to the graph $\tG$ and
the weight function $\tw$. Since each edge of $G$ is considered only
when its $\Theta$-class is contracted, the algorithm has complexity
$O(m)$.

\subsection{The median $\Med_w(G)$}\label{s-algomed}
	
We start with a simple property of the median function $F_w$ that
follows from Lemma~\ref{halfspaces}: 
\begin{lemma}\label{F(x)-F(y)}
  If $xy\in E_i$ with $x\in H'_i$ and $y\in H''_i$, then
  $F_w(x)-F_w(y)=w(H''_i)-w(H'_i)$.
\end{lemma}
	
A halfspace $H$ of $G$ is \emph{majoritary} if $w(H)>\frac{1}{2}w(G)$,
\emph{minoritary} if $w(H)<\frac{1}{2}w(G)$, and \emph{egalitarian} if
$w(H)=\frac{1}{2}w(G)$. Let
$\lMedw(G) = \{v \in V : F_w(v)\le F_w(u), \forall u \sim v\}$ be the
set of local medians of $G$. We continue with the majority
rule: 
\begin{proposition}[\!\!\cite{BaBa,SoCh_Weber}]\label{majority}
  $\Med_w(G)$ is the intersection of all majoritary halfspaces and
  $\Med_w(G)$ intersects all egalitarian halfspaces. If $H'_i$ and
  $H''_i$ are egalitarian halfspaces, then $\Med_w(G)$ intersects both
  $H'_i$ and $H''_i$. Moreover, $\Med_w(G)=\lMedw(G)$.
\end{proposition}
\begin{proof}
  Let us first prove a generalization of Lemma~\ref{F(x)-F(y)} from
  which the different statements of Proposition~\ref{majority} easily
  follow.
  \begin{lemma}\label{majority-claim}
    Let $E_i$ be a $\Theta$-class of a median graph $G$ and let
    $H'_i,H''_i$ be the two halfspaces defined by $E_i$. If
    $x''\in H''_i$ and $x'$ is its gate in $H'_i$, then
    $F_w(x'')\ge F_w(x')+d(x'',x')(w(H'_i)-w(H''_i))$.
  \end{lemma}
		
  \begin{proof}
    By definition of the median function,
    \[
      F_w(x'')-F_w(x') = \sum_{u\in V}{d(x'',u)w(u)}-\sum_{u\in V}
      {d(x',u)w(u)}=\sum_{u\in V}{(d(x'',u)-d(x',u))w(u)}.
    \]
    Then, we decompose the sum over the complementary halfspaces
    $H'_i$ and $H''_i:$
    \[
      F_w(x'')-F_w(x')=\sum_{u'\in
        H_i'}{(d(x'',u')-d(x',u'))w(u')}+\sum_{u''\in
        H_i''}{(d(x'',u'')-d(x',u''))w(u'')}.
    \]
    Since $x'$ is the gate of $x''$ on $H_i'$, for any $u'\in H_i'$,
    $d(x'',u')-d(x',u')= d(x'',x')$. By triangle inequality,
    $d(x'',u'')-d(x',u'')\geq -d(x'',x')$ for every $u''\in H_i''$.
    We get
    \[
      F_w(x'')-F_w(x') \geq \sum_{u \in H_i'}d(x'',x')w(u') -
      \sum_{u''\in H_i''}d(x'',x')w(u'')
    \]
    and conclude that
    $F_w(x'') \geq F_w(x') + d(x'',x')(w(H_i')-w(H_i''))$.
  \end{proof}
  
  Let $H''_i$ and $H'_i$ be two complementary halfspaces such that
  $w(H'_i)>w(H''_i)$. Pick any vertex $x''\in H''_i$ and its gate $x'$
  in $H'_i$. By Lemma~\ref{majority-claim}, $F_w(x'')>F_w(x')$ and
  therefore $x''$ cannot be a median. This shows that the complement
  of a majoritary halfspace does not contain any median vertex. This
  implies that $\Med_w(G)$ is contained in the intersection $M$ of the
  majoritary halfspaces.  If $\Med_w(G)$ is a proper subset of $M$,
  since $M$ is convex we can find two adjacent vertices
  $x\in \Med_w(G)$ and $y\in M\setminus \Med_w(G)$. Let $xy\in E_i$
  with $x\in H'_i$ and $y\in H''_i$. Since $y\in M$, $H'_i$ cannot be
  a majoritary halfspace.  Since $x\in \Med_w(G)$, $H'_i$ cannot be a
  minoritary halfspace. Thus $H'_i$ and $H''_i$ are egalitarian
  halfspaces. Since $F_w(x)-F_w(y)=w(H''_i)-w(H'_i)=0$, we deduce that
  $y$ is a median vertex, thus $\Med_w(G)=M$.
Now, consider two egalitarian complementary halfspaces $H''_i$ and
  $H'_i$. Suppose that a median vertex $x'$ belongs to $H'_i$ and let
  $x''$ be its gate on $H''_i$. By Lemma~\ref{majority-claim},
  $F_w(x'')\le F_w(x')$. Therefore, $x''$ is also median.  By
  symmetry, we conclude that both $H'_i$ and $H''_i$ contain a median
  vertex.

  We now show that any local median is a median. Pick any
  vertex $v\notin \Med_w(G)$. Since $\Med_w(G)$ is the intersection of all majority halfspaces of
  $G$, there exists a majority halfspace $H$ containing $\Med_w(G)$
  and not containing $v$. Let $v'$ be the gate of $v$ in $H$ and $u$
  be a neighbor of $v$ in $I(v,v')$. Then necessarily
  $H\subseteq W(u,v)$, thus $W(u,v)$ is a majoritary halfspace.  This
  implies that $F_w(u)<F_w(v)$, i.e., $v$ is not a local median. This
  concludes the proof of Proposition~\ref{majority}.
\end{proof}

We use Proposition~\ref{majority} and the weights of halfspaces
computed above to derive $\Med_w(G)$. For this, we direct the edges $v'v''$ of each
$\Theta$-class $E_i$ of $G$ as follows. If $v'\in H'_i$ and
$v''\in H''_i$, then we direct $v'v''$ from $v'$ to $v''$ if $w(H''_i)>w(H'_i)$ and
from $v''$ to $v'$ if $w(H'_i)>w(H''_i)$. If $w(H'_i)=w(H''_i)$, then
the edge $v'v''$ is not directed.  We denote this partially directed
graph by $\oG$. A vertex $u$ of $G$ is a \emph{sink} of $\oG$ if there is no edge $uv$
directed in $\oG$ from $u$ to $v$. From Lemma~\ref{F(x)-F(y)}, $u$ is
a sink of $\oG$ if and only if $u$ is a local median of $G$. By
Proposition~\ref{majority}, $\lMedw(G)=\Med_w(G)$ and
thus $\Med_w(G)$ coincides with the set $S(\oG)$ of sinks of $\oG$. Note
that in the graph induced by $\Med_w(G)$, all edges are non-oriented
in $\oG$.
Once all $w(H'_i)$ and $w(H''_i)$ have been computed, the orientation
$\oG$ of $G$ can be constructed in $O(m)$ by traversing all
$\Theta$-classes $E_i$ of $G$. The graph induced by $S(\oG)$ can then
be found in $O(m)$.
\begin{theorem}\label{median}
  The median $\Med_w(G)$ of a median graph $G$ can be computed in
  $O(m)$ time.
\end{theorem}

The next remark follows immediately from the majority rule and the
fast computation of the $\Theta$-classes:

\begin{remark}\label{medmaj}
  Given the median set $\Med_w(G)$ of a median graph $G$, one can find
  all majoritary halfspaces of $G$ in linear time $O(m)$.
\end{remark}

\subsection*{Computing a diametral pair of $\Med_w(G)$}
	
The article~\cite{BaBa} proved that in a median graph, the median set
coincide with the interval between two diametral pairs of its
vertices.  We show how to find this pair in $O(m)$ time, using a
corollary of Proposition~\ref{majority} :
	
\begin{corollary}\label{disjoint-halfspaces}
  If two disjoint halfspaces $H'$ and $H''$ defined by two laminar
  $\Theta$-classes of $G$ both intersect $\Med_w(G)$,
then $w(v)=0$ for any vertex $v\in V\setminus(H'\cup H'')$.
\end{corollary}
	
\begin{proof}
  Since $\Med_w(G)\cap H'\ne \emptyset$ and
  $\Med_w(G)\cap H''\ne \emptyset$, and since $H'$ and $H''$ are
  disjoint, by Proposition~\ref{majority}, both $H'$ and $H''$ are
  egalitarian halfspaces. Since $H'$ and $H''$ are disjoint,
  $w(H')+w(H'')=w(V)$, hence $w(V\setminus (H'\cup H''))=0$.
\end{proof}

\begin{corollary}\label{cor-interval}
  If $w(G)> 0$, we can find $u,v \in V(G)$ in $O(m)$ such that
  $\Med_w(G) = I(u,v)$.
\end{corollary}
	
The proof of Corollary~\ref{cor-interval} is based on a result
of~\cite[Proposition~6]{BaBa} stating that in a median graph, the
median set is an interval.  Let $H$ be a gated subgraph of $G$ and $u$
be a vertex of $H$. The set
$P_{H}(u) = \{ v\in V : u \text{ is the gate of } v \mbox{ in } H \}$
is called the \emph{fiber} of $u$ with respect to $H$. We say that a
fiber $P_{H}(u)$ is \emph{positive} if $w(P_H(u))>0$.  The fibers
$\{P_{H}(u): u\in H\}$ define a partition of $V(G)$. We give below a
non lattice-based proof of the result of Bandelt and
Barth\'el\'emy~\cite[Proposition~6]{BaBa}.
	
\begin{proposition}\label{interval}
  Let $M = \Med_w(G)$ and let $u$ be a vertex of $M$ with a positive
  fiber $P_{M}(u)$.  Then $M=I(u,v)$ for the vertex $v \in M$
  maximizing $d(u,v)$.
\end{proposition}

\begin{proof}
  Since $M$ is convex and $u,v \in M$, we have $I(u,v) \subseteq
  M$. We now prove the reverse inclusion.  Suppose by way of
  contradiction that there exists a vertex $z \in M \setminus I(u,v)$.
  Consider such a vertex $z$ minimizing $d(v,z)$. Let $z'$ be the
  median of $u,v,z$. Then $z'\in I(u,v)$.  Since $M$ is convex and
  $z,z'\in M$, from the minimality choice of $z$ we conclude that $z$
  and $z'$ are adjacent. Since $G$ is bipartite and $v$ is a furthest
  from $u$ vertex of $M$, necessarily $z'\ne v$. Since
  $z\notin I(u,v),z'\in I(u,v)$ and $G$ is bipartite,
  $z'\in I(z,u)\cap I(z,v)$, i.e., $u,v\in W(z',z)$. Let $y$ be any
  neighbor of $z'$ in $I(z',v)$ and suppose that $z'z\in E_i$ and
  $z'y\in E_j$.
		
  We assert that the $\Theta$-classes $E_i$ and $E_j$ are
  laminar. Indeed otherwise, by Lemma~\ref{crossing}, there exists a
  vertex $x' \sim z,y$.
Since $x'\in I(z,y)$ and $z,y\in M$, $x'$ also belongs to $M$ and is
  one step closer to $v$ than $z$.  The minimality choice of $z$
  implies that $x'\in I(v,u)$. Since $u\in W(z,x')$, we have
  $z\in I(x',u)$, yielding $z\in I(v,u)$. This contradiction shows
  that $E_i$ and $E_j$ are laminar, i.e., the halfspaces $W(z,z')$ and
  $W(y,z')$ are disjoint.  Since they both intersect $M$, by
  Corollary~\ref{disjoint-halfspaces}, all vertices of $G$ not
  belonging to $W(z,z')\cup W(y,z')$ have weight 0.
		
  Pick any vertex $p\in P_M(u)$. Since $p$ belongs to the fiber
  $P_M(u)$ and $z,y\in M$, necessarily $u\in I(y,p)\cap I(z,p)$.
  Since $y$ is a neighbor of $z'$ in $I(z',v)$ and $z'\in I(v,u)$, we
  obtain that $z'\in I(y,u)\subseteq I(y,p)$, yielding $p\in W(z',y)$.
  Analogously, from the choice of $z$ we deduced that
  $z'\in I(z,u)\subseteq I(z,p)$, yielding $p\in W(z',z)$. This
  establishes that $P_M(u)\subseteq W(z',y)\cap W(z',z)$, i.e.,
  $P_M(u)$ is disjoint from the halfspaces $W(y,z')$ and $W(z,z')$.
  This contradicts the fact that $P_M(u)$ has positive weight.
\end{proof}
	
Now we can prove the corollary :

\begin{proof}[Proof of Corollary~\ref{cor-interval}]
  Once we have computed $M=\Med_w(G)$, one can pick an arbitrary
  vertex $v_1$ such that $w(v_1) > 0$. By running a BFS-algorithm, we
  find in linear time the closest to $v_1$ vertex $u$ in $M$ and the
  furthest from $v_1$ vertex $v$ in $M$. Since $M$ is gated, $u$ is
  unique, $v_1 \in P_M(u)$ and $v$ is at maximum distance from $u$ in
  $M$. By Proposition~\ref{interval}, $I(u,v) = \Med_w(G)$.
\end{proof}

\subsection*{Median graphs and the majority rule}

In the Introduction we mentioned that the median graphs are the
bipartite graphs satisfying the majority rule. We will make now this
statement precise. Let $G=(V,E)$ be a bipartite graph. A
\emph{halfspace} of $G$ is a subgraph induced by
$W(u,v)=\{ x\in V: d(x,u)<d(x,v)\}$ for some edge $uv$.
Recall that all halfspaces of $G$ are convex iff $G$ is isometrically
embeddable into a hypercube~\cite{Dj}. For a weight function $w$ on
$G$ and any pair of complementary halfspaces $W(u,v)$ and $W(v,u)$, we
have $w(W(u,v))+w(W(v,u))=w(V)$. A \emph{majoritary halfspace} of $G$
is a halfspace $W(u,v)$ such that $w(W(u,v))>w(W(v,u))$. A bipartite
graph $G$ \emph{satisfies the majority rule} if for any weight
function $w$ on $G$, $\Med_w(G)$ is the intersection of all majoritary
halfspaces of $G$.
	
\begin{proposition}\label{prop-maj-median}
  A bipartite graph $G$ satisfies the majority rule if and only if $G$
  is median.
\end{proposition}
	
\begin{proof}
  One direction is covered by Proposition~\ref{majority}. Now, suppose
  that a bipartite graph $G$ satisfies the majority rule. To prove
  that $G$ is median it suffices to show that $G$ satisfies the
  quadrangle condition and does not contain $K_{2,3}$~\cite{Mu}. To
  establish the quadrangle condition, let
  $d(u,z)=k+1, d(u,x)=d(u,y)=k\ge 2$ and $z\sim x,y$.
Consider the weight function $w(u)=w(x)=w(y)=1$ and $w(v)=0$ is
  $v\in V\setminus \{ u,x,y\}$. Notice that
  $F_w(x)=F_w(y)=k+2, F_w(z)=k+3,F_w(u)=2k$ and that $F_w(v)\ge k+1$
  for any other vertex $v$.  Since $W(x,z)$ and $W(y,z)$ are
  majoritary halfspaces and $x\notin W(y,z), y\notin W(x,z)$, the
  vertices $x$ and $y$ are not medians.  This implies that if
  $v\in \Med_w(G)$, then $F_w(v)=k+1$. Since $G$ is bipartite, this is
  possible only if $d(v,u)=k-1, d(v,x)=d(v,y)=1$, i.e., $G$ satisfies
  the quadrangle condition. Suppose now that $G$ contains a $K_{2,3}$
  induced by the vertices $x,y,z,u,u'$, where $u$ and $u'$ are
  adjacent to $x,y,z$. Consider the weight function
  $w(x)=w(y)=w(z)=w(u)=w(u')=1$ and $w(v)=0$ for any vertex
  $v\in V\setminus \{ x,y,z,u,u'\}$.  Then $F_w(u)=F_w(u')=5$ and
  $F_w(x)=F_w(y)=F_w(z)=6$. Since $G$ is bipartite, $F_w(v)\ge 6$ for
  any other vertex $v$. Thus $\Med_w(G)=\{ u,u'\}$.  Since
  $u,y,z\in W_w(u,x)$ and $x,u'\in W(x,u)$ we deduce that $W(u,x)$ is
  a majoritary halfspace and $u'\in \Med_w(G)\setminus W(u,x)$, a
  contradiction with the majority rule.
\end{proof}

\subsection{The Wiener index $W_w(G)$ and the distance matrix $D(G)$
  of $G$.}\label{app-Wiener}
	
\begin{algorithm}[h]
  \caption{DistanceMatrix($G,\Theta$)}
  \label{alg:calcapsp}
  \DontPrintSemicolon \SetAlgoVlined
  \KwData{A median graph $G=(V,E)$, the $\Theta$-classes $\Theta =
    (E_1, \hdots,E_q)$ ordered by increasing distance to the basepoint
    $v_0$.}  
  \KwResult{The distance matrix $D : V\times V \rightarrow \N$}
  \Begin{
    \uIf{$G$ contains a single vertex $v$}{
      $D(v,v) \leftarrow 0$\;
      \Return $D$
    }
    \Else{
      Let $H'$ and $H''$ be two complementary halfspaces defined by
      $E_q$ ($v_0\in H''$)\; 
$D\leftarrow$ DistanceMatrix($H'', \Theta \setminus \{E_q\}$) \;
      \ForEach{$u'u'' \in E_q$ with $u' \in H'$ and $u''\in H''$}
      {
\lForEach{$v''\in H''$}{
          $D(u', v'')\leftarrow D(u'',v'')+1$}
        \lForEach{$v'v'' \in E_q$ with $v'\in H'$ and $v''\in H''$}
        {
$D(u', v')\leftarrow D(u'',v'')$
        }
      }
      \Return $D$
    }
  }
\end{algorithm}

Using the fast computation of the $\Theta$-classes and of the weights
of halfspaces of a median graph $G$, we can compute the Wiener index
of $G$ in linear time.
	
\begin{proposition}\label{t-wiener}
  The Wiener index $W_w(G)$ of a median graph $G$ can be computed in
  $O(m)$ time.
\end{proposition}
	
\begin{proof}
  Given the weights $w(H'_i)$ and $w(H''_i)$ of all halfspaces of $G$,
  the Wiener index $W_w(G)$ of $G$ can be computed in $O(q)$ time
  using the following formula (which holds for all partial cubes, see
  e.g.~\cite{Kl_wiener}):
		
  \begin{lemma}\label{wiener_folklore}
    $W_w(G)=\sum_{i=1}^q w(H'_i)\cdot w(H''_i)$.
  \end{lemma}
  
  \begin{proof}
    If $G$ is a partial cube, then for any two vertices $u,v$ of $G$,
    $d(u,v)$ is equal to the number of $\Theta$-classes $E_i$
    separating $u$ and $v$ ($u$ and $v$ belong to different halfspaces
    defined by $E_i$). We write $u|_iv$ if $E_i$ separates $u$ and
    $v$.  This implies that
    \[
      \sum_{u\in V}\sum_{v\in V} w(u)\cdot w(v)\cdot d(u,v)=\sum_{u\in
        V}\sum_{v\in V}\sum_{i: u|_i v} w(u)\cdot w(v)\cdot
      1=\sum_{i=1}^{q} \sum_{u\in H'_i} \sum_{v\in H''_i} w(u)\cdot w(v)
    \]
    and thus $W_w(G)=\sum_{i=1}^q w(H'_i)\cdot w(H''_i)$.
  \end{proof}
  
  Since the weight of all halfspaces can be computed in $O(m)$ time
  and since $q \leq m$, $W_w(G)$ can also be computed in $O(m)$ time.
\end{proof}

The distance matrix $D(G)$ of a median graph $G$ can be computed in
$O(n^2)$ time by traversing the reverse peripheral peeling
$G_0,\ldots,G_{q-1},G_q=G$ of $G$ (the pseudo-code is given in
Algorithm~\ref{alg:calcapsp}). For each $i$, we compute $D(G_i)$
assuming $D(G_{i-1})$ already computed.  Since $G_{i-1}$ coincides
with the halfspace $H''_i$ of $G_i$, $G_{i-1}$ is gated in $G_i$, thus
$D(G_i)$ restricted to $G_{i-1}$ coincides with $D(G_{i-1})$. Thus, to
obtain $D(G_i)$ we have to compute the distances from all vertices
$v'$ of $H'_i$ to all other vertices of $G_i$.  For each pair $u',v'$
of $H'_i$, let $u'',v''$ be their unique neighbors in $G_{i-1}=H''_i$.
Since $H'_i$ is peripheral in $G_i$, $H'_i$ is isomorphic to the
boundary $\partial H''_i$ of $H''_i$. Since $\partial H''_i$ is gated
(by Lemma~\ref{boundary}), $d(u',v')=d(u'',v'')$. Since $v''$ is the
gate of $v'$ in $H''_i$, for each vertex $w''\in H''_i$ we have
$d(v',w'')=d(v'',w'')+1$. This establishes how to complete the
distance matrix $D(G_i)$ from $D(G_{i-1})$. The complexity of this
completion is the number of items of $D(G_i)$ which are not in
$D(G_{i-1})$. This shows that $D(G)$ can be computed in total $O(n^2)$
time. Consequently, we obtain the following result:
	
\begin{proposition}\label{t-dist}
  The distance matrix $D(G)$ of a median graph $G$ can be computed in
  $O(n^2)$ time.
\end{proposition}

\section{The median problem in the cube complex of
  $G$}\label{sec-cubecomplex}

\subsection{The main result}\label{s-problegeom}
In this section, we describe a linear time algorithm to compute
medians in cube complexes of median graphs.
	
\subsection*{The problem}
Let $G=(V,E)$ be a median graph with $n$ vertices, $m$ edges, and $q$
$\Theta$-classes $E_1,\ldots, E_q$.  Let $\G$ be the \emph{cube
  complex} of $G$ obtained by replacing each graphic cube of $G$ by a
unit solid cube and by isometrically identifying common subcubes. We
refer to $\G$ as to the \emph{geometric realization} of $G$ (see
Fig.~\ref{the_complex}(a)). We suppose that $\G$ is endowed with the
intrinsic $\ell_1$-metric $d_1$.  Let $P$ be a finite set of points of
$(\G,d_1)$ (called \emph{terminals}) and let $w$ be a weight function
on $\G$ such that $w(p)>0$ if $p\in P$ and $w(p)=0$ if $p\notin
P$. The goal of the \emph{median problem} is to compute the set
$\Med_w(\G)$ of median points of $\G$, i.e., the set of all points
$x\in \G$ minimizing the function
$F_w(x)=\sum_{p\in \G} w(p)d_1(x,p)=\sum_{p\in P} w(p)d_1(x,p)$.
	
\subsection*{The input}
The cube complex $\G$ is given  by its 1-skeleton $G$.
Each terminal $p\in P$ is given by its coordinates in the smallest
cube $Q(p)$ of $\G$ containing $p$. Namely, we give a vertex $v(p)$ of
$Q(p)$ together with its neighbors in $Q(p)$ and the coordinates of $p$ in
the embedding of $Q(p)$ as a unit cube in which $v(p)$ is the origin
of coordinates.
Let $\delta$ be the sum of the sizes of the encodings of the
points of $P$. Thus the input of the median problem has size
$O(m+\delta)$.
	
\subsection*{The output} Unlike $\Med_w(G)$ (which is a gated subgraph
of $G$), $\Med_w(\G)$ is not a subcomplex of $\G$. Nevertheless we
show that $\Med_w(\G)$ is a subcomplex of the box complex $\hG$
obtained by subdividing $\G$, using the hyperplanes passing via the
terminals of $P$.  The output is the 1-skeleton $\wM$ of $\Med_w(\wG)$ and $\Med_w(\hG)$, and the local
coordinates of the vertices of $\wM$ in $\G$. We show that the output
has linear size $O(m)$.

\begin{theorem}\label{mediancomplexx}
  Let $G$ be a median graph with $m$ edges and let $P$ be a finite set
  of terminals of $\G$ described by an input of size $\delta$.  The
  1-skeleton $\wM$ of $\Med_w(\G)$ can be computed in linear time
  $O(m+\delta)$.
\end{theorem}
	
\begin{figure}[t]
  \captionsetup[subfigure]{singlelinecheck=true}
  \centering
\subcaptionbox{\label{fig:Complex}}{\includegraphics[scale=0.64,page=5]{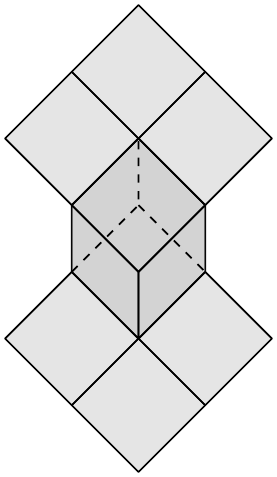}}
  \qquad\qquad \subcaptionbox{\label{fig:hyperplane}}
  {\includegraphics[scale=0.64,page=4]{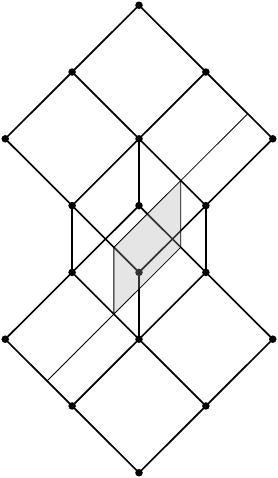}}
  \qquad\qquad \subcaptionbox{\label{fig:boxcomplex}}
  {\includegraphics[scale=0.64,page=6]{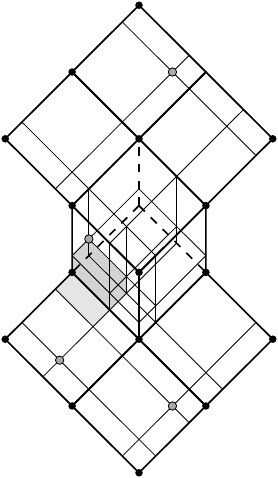}}
  \caption{~(a) The cube complex ${\mathcal D}$ of $D$,~(b) a
    hyperplane of $\mathcal D$, and~(c) the box complex
    $\widehat{\mathcal D}$ and $\Med_w({\mathcal D})$ (in gray)
    defined by 4 terminals of weight $1$.}\label{the_complex}
\end{figure}

\subsection{Geometric halfspaces and hyperplanes}\label{ss-geomhalfs}

In the following, we fix a basepoint $v_0$ of $G$. For each point $x$
of $\G$, let $Q(x)$ be the smallest cube of $\G$ containing $x$ and
let $v(x)$ be the gate of $v_0$ in $Q(x)$. For each $\Theta$-class
$E_i$ defining a dimension of $Q(x)$, let $\epsilon_i(x)$ be the
coordinate of $x$ along $E_i$ in the embedding of $Q(x)$ as a unit
cube in which $v(x)$ is the origin.
For a $\Theta$-class $E_i$ and a cube $Q$ having $E_i$ as a dimension,
the \emph{$i$-midcube} of $Q$ is the subspace of $Q$ obtained by
restricting the $E_i$-coordinate of $Q$ to $\frac{1}{2}$.
A \emph{midhyperplane} $\ch_i$
of $\G$
is the union of all $i$-midcubes.
Each $\ch_i$ cuts $\G$ in two components~\cite{Sa} and the union of
each of these components with $\ch_i$ is called a \emph{geometric
  halfspace} (see Fig.~\ref{the_complex}(b)).
The \emph{carrier} $\cN_i$ of $E_i$ is the union of all cubes of $\G$
intersecting $\ch_i$; $\cN_i$ is isomorphic to $\ch_i\times[0,1]$. For
a $\Theta$-class $E_i$ and $0< \epsilon < 1$, the \emph{hyperplane}
$\ch_i(\epsilon)$ is the set of all points $x\in \cN_i$ such that
$\epsilon_i(x)=\epsilon$. Let $\ch_i(0)$ and $\ch_i(1)$ be the
respective geometric realizations of $\partial H_i''$ and
$\partial H_i'$. Note that $\ch_i(\epsilon)$ is obtained from $\ch_i$
by a translation.
The \emph{open carrier} $\cN^{\circ}_i$ is
$\cN_i\setminus(\ch_i(0)\cup \ch_i(1))$.
We denote by $\cH'_i(\epsilon)$ and $\cH''_i(\epsilon)$ the geometric
halfspaces of $\G$ defined by
$\ch_i(\epsilon)$. Let $\cH''_i:=\cH''_i(0)$ and $\cH'_i:=\cH'_i(1)$; they are the
geometric realizations of $H'_i$ and $H''_i$.  Note that $\G$ is the
disjoint union of $\cH'_i$, $\cH''_i$, and $\cN^{\circ}_i$.
	
\subsection{The majority rule for $\G$}
	
Now we show how to reduce the median problem in $\G$ to a median problem
in a median graph.

\subsection*{The box complex $\hG$}
By~\cite[Theorem 3.16]{vdV}, $(\G,d_1)$ is a median metric space
(i.e., $|I(x,y)\cap I(y,z)\cap I(z,x)|=1$ $\forall x,y,z\in \G$) and
the graph $G$ is isometrically embedded in $(\G,d_1)$.  For each
$p\in P$ and each coordinate $\epsilon_i(p)$,
consider the hyperplane $\ch_i(\epsilon_i(p))$. All such hyperplanes
subdivide $\G$ into a box complex $\hG$ (see
Fig.~\ref{the_complex}(c)).
Clearly, $(\hG,d_1)$ is a median space.  By~\cite[Theorem 3.13]{vdV},
the 1-skeleton $\wG$ of $\hG$ is a median graph and each point of $P$
corresponds to a vertex of $\wG$. The $\Theta$-classes of $\wG$ are
subdivisions of the $\Theta$-classes of $G$. In $\hG$, all edges of a
$\Theta$-class of $\wG$ have the same length. Let $\wG_l$ be the graph
$\wG$ in which the edges have these lengths. $\wG_l$ is a median
space, thus $\Med_w(\wG_l)=\Med_w(\wG)$ by~\cite{SoCh_Weber}.  By
Proposition~\ref{majority}, $\Med_w(\wG_l)$
is the intersection of the majoritary halfspaces of
$\wG$.

\begin{proposition}\label{med_as_subcomplex}
  $\Med_w(\G)$ is the subcomplex of $\hG$ defined by
  $\wM:=\Med_w(\wG_l)$.
\end{proposition}
	
\begin{proof}
  Let $\hE_1, \ldots, \hE_{\hat{q}}$ be the the
  $\Theta$-classes of $\wG$. For a point $x \in
  \G$, we denote by $\hB(x)$ the smallest box of $\hG$ containing
  $x$, and for any $\Theta$-class $\hE_i$ of
  $\hB(x)$, let $\hepsilon_i(x)$ be the coordinate of
  $x$ along the dimension $\hE_i$ in the embedding of
  $\hB(x)$ as a unit cube where the origin is the gate of
  $v_0$ on $\hB(x)$.
		
  For a point $x\in \G$, let $\hG(x)$ be the subdivision of the
  complex $\hG$ by the hyperplanes passing via $x$, let $\wG(x)$
  denote the 1-skeleton of $\hG(x)$ and let $\wG_l(x)$ be the
  corresponding weighted graph. Again, $\wG(x)$ is a median graph and
  $\wG_l(x)$ is an isometric subgraph of $\hG(x)$ that is isometric to
  $\G$ and $\hG$.
		
  We show by induction on the dimension of $\hB(x)$ that
  $x \in \Med_w(\G) = \Med_w(\hG) = \Med_w(\hG(x))$ if and only if for
  any vertex $v$ of $\hB(x)$, $v \in \Med_w(\G)$. If $\hB(x) = \{x\}$,
  then there is nothing to prove. Otherwise, pick a $\Theta$-class
  $\hE_i$ of $\wG$ such that $0 < \hepsilon_i(x) <1$. In $\wG(x)$, $x$
  has exactly two neighbors $x',x''$ belonging to two opposite facets
  of $\hB(x)$ such that $\hepsilon_i(x'')=0$ and $\hepsilon_i(x')=1$.
  Observe that $V(\hB(x)) = V(\hB(x')) \cup V(\hB(x''))$ and
  $x \in I_{\wG(x)}(x',x'')$.  By the definition of $\wG$, there is no
  terminal $p \in P$ with $0 <\hepsilon_i(p) <1$. Consequently in
  $\wG(x)$, $W(x'',x) \cap P = W(x,x') \cap P$ and
  $W(x',x) \cap P = W(x,x'') \cap P$. Therefore in $\wG(x)$ (and in
  $\wG_l(x)$), the halfspace $W(x'',x)$ is majoritary
  (resp.\ egalitarian, minoritary) if and only if $W(x',x)$ minoritary
  (resp.\ egalitarian, majoritary).
		
  Suppose that $x \in \Med_w(\G)$. If $W(x,x'')$ (resp.\ $W(x,x')$) is
  minoritary in $\wG(x)$, then by Lemma~\ref{F(x)-F(y)} applied to
  $\wG_l(x)$, $F_w(x) > F_w(x'')$ (resp.\ $F_w(x) > F_w(x')$) and
  $x \notin \Med_w(\G)$, a contradiction. Thus, necessarily $W(x'',x)$
  and $W(x',x)$ are egalitarian and by Lemma~\ref{F(x)-F(y)} in
  $\wG_l(x)$, $F_w(x')=F_w(x'')=F_w(x)$. Since $\hB(x')$ and
  $\hB(x'')$ are facets of $\hB(x)$, by induction hypothesis, all
  vertices in $V(\hB(x)) = V(\hB(x')) \cup V(\hB(x''))$ belong to
  $\Med_w(\G)$. Conversely, suppose that
  $V(\hB(x)) = V(\hB(x')) \cup V(\hB(x'')) \subseteq \Med_w(\G)$. Then
  by induction hypothesis, $x',x'' \in \Med_w(\G)$. Since
  $\Med_w(\wG_l(x)) = \Med_w(\wG(x))$ is convex and
  $x \in I_{\wG(x)}(x,x')$, we have $F_w(x) = F_w(x') = F_w(x'')$ and
  consequently, $x \in \Med_w(\G)$.
\end{proof}

\subsection*{The $E_i$-median problems}
We adapt now Proposition~\ref{majority} to the continuous setting. In
our algorithm and next results we will not explicitly construct the
box complex $\hG$ and its 1-skeleton $\widehat{G}$ (because they are
too large), but we will only use them in proofs.
For a $\Theta$-class of $G$, the $E_i$-\emph{median} is the median of
the multiset of points of the segment $[0,1]$ weighted as follows: the
weight $w_i(0)$ of $0$ is $w(\cH''_i)$, the weight $w_i(1)$ of 1 is
$w(\cH'_i)$, and for each $p\in P\cap \cN_i^{\circ}$, there is a point
$\epsilon_i(p)$ of $[0,1]$ of weight $w_i(\epsilon_i(p))=w(p)$.  It is
well-known that this median is a segment $[\varrho''_i,\varrho'_i]$
defined by two consecutive points $\varrho''_i\le \varrho'_i$ of
$[0,1]$ with positive weights,
and for any $p \in P$, $\epsilon_i(p)\leq \varrho_i''$ or
$\epsilon_i(p) \geq \varrho_i'$. \emph{Majoritary}, \emph{minoritary},
and \emph{egalitarian} geometric halfspaces of $\G$ are defined in the
same way as the halfspaces of $G$.

\begin{proposition}\label{majority_bis}
  Let $E_i$ be a $\Theta$-class of $G$.  Then the following holds:
  \begin{enumerate}
  \item $\Med_w(\G)\subseteq \cH''_i$ (resp.\ 
    $\Med_w(\G)\subseteq \cH'_i$) if and only if $\cH''_i$ is
    majoritary (resp.\ $\cH'_i$ is majoritary), i.e.,
    $\rho_i'' = \rho_i' = 0$ (resp.\ $\rho_i'' = \rho_i' = 1$);
			
  \item $\Med_w(\G)\subseteq \cH''_i\cup \cN^{\circ}_i$ (resp.\ 
    $\Med_w(\G)\subseteq \cH'_i\cup \cN^{\circ}_i$) and $\Med_w(\G)$
    intersects each of the sets $\cH''_i$ (resp.\ $\cH'_i$) and
    $\cN^{\circ}_i$ if and only if $\cH''_i$ (resp.\ $\cH'_i$) is
    egalitarian and $\cH'_i$ (resp.\ $\cH''_i$) is minoritary, i.e.,
    $0 = \rho_i'' < \rho_i' < 1$ (resp.\ $0 < \rho_i'' < \rho_i' = 1$);

  \item $\Med_w(\G)\subseteq \cN^{\circ}_i$ if and only if $\cH'_i$
    and $\cH''_i$ are minoritary, i.e.,
    $0 < \rho_i'' \leq \rho_i' < 1$;
			
  \item $\Med_w(\G)$ intersects the three sets $\cH_i, \cH''_i,$ and
    $\cN^{\circ}_i$ if and only if $\cH'_i$ and $\cH''_i$ are
    egalitarian, i.e., $0 = \rho_i'' \leq \rho_i' = 1$ (and thus
    $w(\cN^{\circ}_i)=0$).
  \end{enumerate}
\end{proposition}

\begin{proof}
  Let $0<\epsilon_1<\cdots<\epsilon_k<1$ denote the possible values of
  coordinates of points of $P$ with respect to the $\Theta$-class
  $E_i$. They define parallel hyperplanes
  $\ch_i(0),\ch_i(\epsilon_1),\ldots,\ch_i(\epsilon_k),\ch_i(1)$. The
  pieces of the edges of $E_i$ bounded by two such consecutive
  hyperplanes define a $\Theta$-class of $\widehat{G}$ and all such
  $\Theta$-classes are laminar. In fact, we have the following chains
  of inclusions between the geometric halfspaces of $\G$ (or $\hG$)
  defined by those $\Theta$-classes:
  $\cH''_i=\cH_i''(0)\subset \cH''_i(\epsilon_1)\subset\cdots\subset
  \cH''_i(\epsilon_k)$ and
  $\cH'_i=\cH'_i(1)\subset \cH'_i(\epsilon_k)\subset\cdots\subset
  \cH'_i(\epsilon_1)$. Similar inclusions hold between corresponding
  halfspaces of the graph $\widehat{G}$. Notice also that, by the
  definition of $\widehat{G}$ and $\hG$, the geometric halfspaces and
  the corresponding graphic halfspaces, have the same weight.
  Therefore, to deduce the different cases of the proposition, it
  suffices to apply the majority rule (Proposition~\ref{majority} and
  Corollary~\ref{disjoint-halfspaces}) to the halfspaces of
  $\widehat{G}$ occurring in the two chains of inclusions and use the
  fact that $\wM$ is the $1$-skeleton of $\Med_w(\G) = \Med_w(\hG)$ in
  $\hG$ from Proposition~\ref{med_as_subcomplex}.
\end{proof}
	
\subsection{The algorithm}
	
\begin{algorithm}[]
  \caption{ComputeMedianCubeComplex($G,P,w,\Theta$)}\label{alg:compmedcomplex}
\SetKw{Add}{add}
  \SetKw{To}{to}
  \SetKwFunction{ComputeWeightsOfHalfspaces}{ComputeWeightsOfHalfspaces}
  \DontPrintSemicolon
  \SetAlgoVlined
  \KwData{A median graph $G=(V,E)$, a set of terminals $P$, a weight
    function $w : P\rightarrow \R^+$, the $\Theta$-classes $\Theta =
    (E_1, \hdots,E_q)$ of $G$ ordered by increasing distance to the
    basepoint $v_0$.} 
  \KwResult{The graph $\wM$ and the coordinates of each vertex $\hm \in
    \wM$ in $\G$}
  \Begin{
    Modify the root $v(p)$ of each point $p \in P$ such that $v(p)$
    is the gate of $v_0$ on $Q(p)$ \;
    Compute  $w(P_i)$ for all $\Theta$-classes $E_i$ by
    traversing $P$\;
Compute  $w_*(v)=w(P_v)$ for all $v \in V$ by traversing $P$ \;
    Apply \ComputeWeightsOfHalfspaces$(G,w_*,\Theta)$ to compute the
    weights $w_*(H_i')$ and $w_*(H_i'')$ for each $\Theta$-class $E_i$\;
    Set $w(\cH_i') \leftarrow w_*(H_i')$ and $w(\cH_i'') \leftarrow
    w_*(H_i'') - w(P_i)$  for each $\Theta$-class $E_i$\;
    Compute the $E_i$-median instance for all $\Theta$-classes $E_i$ by
    traversing  $P$\;
    Compute the $E_i$-median $[\varrho_i'',\varrho_i']$  for each
    $\Theta$-class $E_i$ \;
    Orient the edges of $G$ and compute the set of half-edges around
    each vertex $v \in V$\;
    Compute the set $S(\oG)$ of sinks of $\oG$ by traversing the
    edges of $G$\;
    $V(\wM) \leftarrow \left\{g(v):v \in S(\oG)\right\}$\;
    $E(\wM)n\leftarrow \left\{g(u)g(v):uv \in E \text{ and } u,v \in
      S(\oG)\right\}$ \;
    \Return $\left(V(\wM),E(\wM)\right)$ \;
  }
\end{algorithm}
	
\subsection*{Preprocessing the input}
We first compute the $\Theta$-classes $E_1, E_2, \ldots, E_q$ of $G$
ordered by increasing distance from $v_0$ to $H'_i$. Using this, we
first modify the input of the median problem in linear time
$O(m + \delta)$ in such a way that for each terminal $p \in P$, $v(p)$
is the gate of $v_0$ in $Q(p)$.  Once the $\Theta$-classes have been
computed, we can assume that each terminal $p$ is described by its
root $v(p)$ as well as a list of coordinates $\Delta(p)$, one
coordinate $0 <\epsilon_i(p)<1$ for each $\Theta$-class $E_i$ of
$Q(p)$ such that $\epsilon_i(p)$ is the coordinate of $x$ along $E_i$
in the embedding of $Q(p)$ as a unit cube in which $v(p)$ is the
origin.  To update $v(p)$ and $\Delta(p)$, we use a (non-initialized)
matrix $B$ whose rows and columns are indexed respectively by the
vertices and the $\Theta$-classes of $G$ and such that if a vertex $v$
has a neighbor $v'$ such that $vv'$ belongs to the $\Theta$-class
$E_i$, then $B[v,E_i] = v'$ (and $B[v,E_i]$ is undefined if $v$ does
not have such a neighbor $v'$). One can construct $B$ in time $O(m)$
by traversing all edges of $G$ once the $\Theta$-classes have been
computed. With the matrix $B$ at hand, for each terminal $p \in P$, we
consider the coordinates of $p$ in order and for each coordinate
$\epsilon_i(p)$ of $p$, if $v'=B[v(p),E_i]$ is closer to $v_0$ than
$v(p)$, we replace $v(p)$ by $v'$ and $\epsilon_i(p)$ by
$1-\epsilon_i(p)$. Observe that each time we modify $v(p)$, $v(p)$ is
still a vertex of $Q(p)$ and thus $B[v(p),E_j]$ is still defined for
any coordinate $\epsilon_j \in \Delta(p)$.  Note that $v(p)$ can move
to distance up to $|\Delta(p)|$ from its original position during the
process. Once the matrix $B$ has been computed, the modification of
the roots of all the points $p \in P$ can be performed in time
$O(\sum_{p\in P} \Delta(p)) = O(\delta)$.
	
In this way, the local coordinates of the terminals of $P$ coincide
with the coordinates $\epsilon_i(p)$ defined in
Section~\ref{ss-geomhalfs}.
For each $\Theta$-class $E_i$, let
$P_i =P\cap \cN^{\circ}_i = \{p \in P : 0 < \epsilon_i(p) <1\}$, and
for each point $v \in V(G)$, let $P_v = \{p \in P : v(p) =v\}$.  By
traversing the points of $P$, we can compute all sets $P_i$,
$1\leq i \leq q$ and $P_v$, $v \in V$ and the weights of these sets in
time $O(\delta)$.
	
\subsection*{Computing the $E_i$-medians}
We first compute the weights $w_i(0) = w(\cH_i'')$ and
$w_i(1) = w(\cH_i')$ of the geometric halfspaces $\cH_i'', \cH_i'$ of
$G$. For each vertex $v$ of $G$, let $\uw(v) = w(P_v)$. Note that
$\uw(V) = w(P)$. Since $v_0 \in H_i''$, $\uw(H_i') = w(\cH_i')$ and
$\uw(H_i'') = w(\cH_i'') + w(\cN_i^o)$ for each $\Theta$-class
$E_i$. We apply the algorithm of Section~\ref{ssec:weight-halfpaces}
to $G$ with the weight function $\uw$ to compute the weights
$\uw(H_i')$ and $\uw(H_i'')$ of all halfspaces of $G$. Since
$w(\cN_i^o) = w(P_i)$ is known, we can compute $w(\cH_i') = \uw(H_i')$
and $w(\cH_i'') = \uw(H_i'') - w(P_i)$.  This allows us to complete
the definition of each $E_i$-median problem which altogether can be
solved linearly in the size of the input~\cite[Problem 9.2]{Cormen},
i.e., in time $O(\Sigma_{i=1}^q (|P_i| +2)) = O(\delta+m)$.
	
\subsection*{Computing $\wM$}
To compute the 1-skeleton $\wM$ of $\Med_w(\G)$ in $\wG$, we orient
the edges of $E_i$ according to the weights of $\cH'_i$ and $\cH''_i$:
$v'v''\in E_i$ with $v'\in \cH'_i$ and $v''\in \cH''_i$ is directed
from $v''$ to $v'$ if $\varrho'_i=\varrho''_i=1$ ($\cH'_i$ is
majoritary) and from $v'$ to $v''$ if $\varrho'_i=\varrho''_i=0$
($\cH''_i$ is majoritary), otherwise the edges of $E_i$ are not
oriented. Denote this partially directed graph by $\oG$ and let
$S(\oG)$ be the set of sinks of $\oG$. A non-directed edge
$v'v'' \in E_i$ defines a \emph{half-edge with origin} $v''$ if
$\varrho''_i>0$ and a \emph{half-edge with origin} $v'$ if
$\varrho'_i < 1$ (an edge $v'v''$ such that
$0< \varrho''_i\le \varrho'_i< 1$ defines two half-edges).
	
\begin{proposition}\label{one_sink}
  For any vertex $v$ of $\oG$, all half-edges with origin $v$ define a
  cube $Q_v$ of $\G$.
\end{proposition}
	
\begin{proof}
For any vertex $v$ and two $\Theta$-classes $E_i,E_j$ defining
  half-edges with origin $v$, let $v_i$ and $v_j$ be the respective
  neighbors of $v$ in $\wG$ along the directions $E_i$ and $E_j$.  By
  Proposition~\ref{majority_bis}, $vv_i$ and $vv_j$ point to two
  majoritary halfspaces of $\wG$ (and $\G$). Since those two
  halfspaces cannot be disjoint, $E_i$ and $E_j$ are crossing.
The proposition then follows  from Lemma~\ref{crossing}.
\end{proof}

For any cube $Q$ of $\G$, let $B(Q) \subseteq Q$ be the subcomplex of
$\hG$ that is the Cartesian product of the $E_i$-medians
$[\varrho''_i,\varrho'_i]$ over all $\Theta$-classes $E_i$ wich define
dimensions of $Q$. By the definition of the $E_i$-medians, $B(Q)$ is a
single box of $\hG$ and its vertices belong to $\wG$.
	
\begin{proposition}
  For any cube $Q$ of $\G$, if $Q \cap \Med_{w}(\G) \neq \emptyset$,
  then $B(Q) = \Med_w(\G) \cap Q$.
\end{proposition}
	
\begin{proof}
  If a vertex $x$ of $B(Q)$ is not a median of $\wG$, by
  Proposition~\ref{majority}, $x$ is not a local median of $\wG$. Thus
  $F_w(x)>F_w(y)$ for an edge $xy$ of $\wG$. Suppose that $xy$ is
  parallel to the edges of $E_i$ of $G$.  Then $\epsilon_i(x)$
  coincides with $\varrho''_i$ or $\varrho'_i$. Since $F_w(x)>F_w(y)$,
  the halfspace $W(y,x)$ of $\wG$ is majoritary, contrary to the
  assumption that $\epsilon_i(x)$ is an $E_i$-median point.  Thus all
  vertices of $B(Q)$ belong to $\wM$ and by
  Proposition~\ref{med_as_subcomplex}, $B(Q)\subseteq \Med_w(\G)$. It
  remains to show that any point of $Q\setminus B(Q)$ is not
  median. Otherwise, by Proposition~\ref{med_as_subcomplex} and since
  $\wM$ is convex, there exists a vertex $y\notin B(Q)$ of
  $(\wM\cap Q)\setminus B(Q)$ adjacent to a vertex $x$ of $B(Q)$. Let
  $xy$ be parallel to $E_i$.  Then $\epsilon_i(x)$ coincides with
  $\varrho''_i$ or $\varrho'_i$ and $\epsilon_i(y)$ does not belong to
  the $E_i$-median $[\varrho''_i,\varrho'_i]$.  Hence the halfspace
  $W(y,x)$ of $\wG$ is minoritary, contrary to $F_w(y)=F_w(x)$.
\end{proof}
	
For a sink $v$ of $\oG$, let $g(v)$ be the point of $Q_v$ such that
for each $\Theta$-class $E_i$ of $Q_v$, $\epsilon_i(g(v)) = \varrho'$
if $v \in \cH_i'$ and $\epsilon_i(g(v)) = \varrho''$ if
$v \in \cH_i''$. Note that $g(v)$ is the gate of $v$ in $B(Q_v)$ and
$g(v)$ is a vertex of $\wM$.  Conversely, let $x \in \wM$ and consider
the cube $Q(x)$. Since $B(Q(x))$ is a cell of $\hG$, for each
$\Theta$-class $E_i$ of $Q(x)$, we have
$\epsilon_i(x) \in \{\varrho'_i,\varrho_i''\}$. Let $f(x)$ be the
vertex of $Q(x)$ such that $f(x) \in \cH_i''$ if
$\epsilon_i(x) = \varrho''_i$ and $f(x) \in \cH_i'$ otherwise.
	
\begin{proposition}\label{prop-sources}
  For any $v\in S(\oG)$, $g(v)$ is the gate of $v$ in $\wM$ and
  $\Med_w(\G)$.  For any $x\in \wM$, $x = g(f(x))$ is the gate of
  $f(x)$ in $\wM$ and $\Med_w(\G)$.
		
  Furthermore, for any edge $uv$ of $G$ with $u,v\in S(\oG)$, either
  $g(u) = g(v)$ or $g(u)g(v)$ is an edge of $\wM$. Conversely, for any
  edge $xy$ of $\wM$, $f(x)f(y)$ is an edge of $G$.
\end{proposition}

\begin{proof}By Proposition~\ref{majority_bis} applied to $\G$,
  Proposition~\ref{majority} applied to $\wG$, and the definition of
  sinks of $\owG$, $g(v)$ is a sink of $\owG$, thus $g(v)$ is a median
  of $\wG$ and $\G$.  Since $B(Q_v) = \Med_w(\G) \cap Q_v$ is gated
  and non-empty, the gate of $v$ in $\Med_w(\G)$ belongs to $B(Q_v)$
  and thus the gate of $v$ in $\Med_w(\G)$ is the gate of $v$ on
  $B(Q_v)$.  Conversely, $\epsilon_i(x) \notin \{0,1\}$ for any $E_i$
  defining a dimension of $Q(x)$, thus there is an $E_i$-half-edge
  with origin $f(x)$.  Pick now any $E_j$-edge incident to $v$ such
  that $E_j$ does not define a dimension of $Q(x)$. Without loss of
  generality, assume that $f(x)\in \cH_j'$. Then $x \in \cH_j'$,
  yielding $w(\cH_j') \geq \frac{1}{2}w(P)$. By
  Proposition~\ref{majority_bis}, $\varrho_j'=1$ and thus $f(x)$ is
  not the origin of an $E_j$-edge or $E_j$-half-edge. Consequently,
  $Q_{f(x)} = Q(x)$ by Proposition~\ref{one_sink} and by the
  definition of $f(x)$ and $g(f(x))$, we have $x= g(f(x))$.
		
  Let $v'v''$ be an $E_i$-edge between two sinks of $\oG$ with
  $v' \in \cH_i'$ and $v'' \in \cH''_i$. Let $x' = g(v')$ and
  $x'' = g(v'')$ and assume that $x' \neq x''$.  Let $u', u''$ be the
  points of $v'v''$ such that $\epsilon_i(u') = \varrho_i'$ and
  $\epsilon_i(u'') = \varrho_i''$. Note that $u'$ and $u''$ are
  adjacent vertices of $\wG$ and that $u' \in I_{\wG}(v',x')$ and
  $u'' \in I_{\wG}(v'',x'')$. In $\wG$, $x''$ is the gate of $u''$
  (and $x'$ is the gate of $u'$) in $\wM$.  Since
  $d_{\wG}(u',x') + d_{\wG}(x',x'') = d_{\wG}(u',x'') \leq
  d_{\wG}(u'',x'') +1$ and
  $d_{\wG}(u'',x'') + d_{\wG}(x',x'') = d_{\wG}(u'',x') \leq
  d_{\wG}(u',x') +1$, we obtain that $d_{\wG}(x',x'') \leq 1$.
		
  Any edge $x'x''$ of $\wM$ is parallel to a $\Theta$-class $E_i$ of
  $G$.  For any $\Theta$-class $E_j$ of $Q(x')$ (resp.\ $Q(x'')$) with
  $j \neq i$, $E_j$ is a $\Theta$-class of $Q(x'')$ (resp.\ $Q(x')$)
  and $\epsilon_j(x') = \epsilon_j(x'')$. By their definition, $f(x')$
  and $f(x'')$ can be separated only by $E_i$, i.e.,
  $d_G(f(x'),f(x'')) \leq 1$. Since $f$ is an injection from $V(\wM)$
  to $S(\oG)$, necessarily $f(x')$ and $f(x'')$ are adjacent.
\end{proof}
		
The algorithm computes the set $S(\oG)$ of all sinks of $\oG$ and for
each sink $v\in S(\oG)$, it computes the gate of $g(v)$ of $v$ in
$\wM$ and the local coordinates of $g(v)$ in $\G$. The algorithm
returns $\left\{g(v): v\in S(\oG)\right\}$ as
$V(\wM)$ and $\left\{g(u)g(v) : uv \in E \text{ and } u,v \in
  S(\oG)\right\}$ as
$E(\wM)$. Proposition~\ref{prop-sources} implies that $V(\wM)$ and
$E(\wM)$ are correctly computed and that $\wM$ contains at most $n$
vertices and $m$ edges. Moreover each vertex $x$ of $\wM$ is the gate
$g(f(x))$ of the vertex $f(x)$ of $Q(x)$ that has dimension at most
$\deg(f(x))$. Hence the size of the description of the vertices of
$\wM$ is at most $O(m)$. This finishes the proof of
Theorem~\ref{mediancomplexx}.

\subsection{Wiener index in $\G$}
We describe a linear time algorithm to compute the Wiener index of a
set of terminals in the $\ell_1$-cube complex $\G$ of a median graph
$G$.  By analogy with graphs, the Wiener index in $\G$ is the sum of
the weighted distances between all pairs of terminals.
	
\begin{proposition}
  Let $G$ be a median graph with $m$ edges and let $P$ be a finite set
  of terminals of $\G$ described by an input of size $\delta$. The
  Wiener index of $P$ in $\G$ can be computed in $O(m+\delta)$ time.
\end{proposition}

\begin{proof}
  The proof is similar to the proof of Proposition~\ref{t-wiener}.
  Let $0<\epsilon_1<\cdots<\epsilon_k<1$ denote the $E_i$-coordinates
  of the points in $P$ and let $\epsilon_0=0$ and $\epsilon_{k+1}=1$.
  Just like in the proof of Proposition~\ref{majority_bis}, we have
  the following chains of inclusions between the halfspaces defined by
  the hyperplanes
  $\ch_i(\epsilon_0),\ch_i(\epsilon_1),\ldots,\ch_i(\epsilon_k),\ch_i
  (\epsilon_{k+1})$:
  \[ \cH''_i=\cH_i''(\epsilon_0)\subset
    \cH''_i(\epsilon_1)\subset\cdots\subset \cH''_i(\epsilon_k) \]
  and 
  \[\cH'_i=\cH'_i(\epsilon_{k+1})\subset
    \cH'_i(\epsilon_k)\subset\cdots\subset \cH'_i(\epsilon_1).\] 
  
  Then, similarly to Lemma~\ref{wiener_folklore}, we get the following
  result:
  \begin{lemma}
    $W_w(\G)=\sum_{i=1}^{q} \sum_{j=0}^k w(\cH''_i(\epsilon_j))\cdot
    w(\cH'_i(\epsilon_{j+1}))\cdot(\epsilon_{j+1}-\epsilon_j)$.
  \end{lemma}
  Once the $O(\delta)$ hyperplanes are ordered, we can compute the
  weights of the halfspaces in $O(m+\delta)$ time and compute the
  Wiener index of $P$ in $\G$ in $O(m+\delta)$ time.
\end{proof}
	
\section{The median problem in event structures}\label{sec:compact}
	
In this section, we consider the median problem in which the median
graph is implicitly defined as the domain of configurations of an
event structure. We show that the problem can be solved efficiently in
the size of the input. However, if the input consists solely of the
event structure and the goal is to compute the median of all
configurations of the domain, then this algorithmic problem is
$\#$P-hard. To prove this we provide a direct (polynomial size)
correspondence between event structures and 2-SAT formulas and use
$\#$P-hardness of a similar median problem for 2-SAT established
in~\cite{Fe}.

\subsection{Definitions and bijections}
We start with the definition of event structures and 2-SAT formulas
and their bijections with median graphs.

\subsubsection{Event structures}
Event structures, introduced by Nielsen, Plotkin, and
Winskel~\cite{NiPlWi,Winskel}, are a widely recognized abstract model
of concurrent computation. An \emph{event structure} is a triple
${\E}=(E,\le, \#)$, where
\begin{itemize}
\item $E$ is a set of \emph{events},
\item $\le\subseteq E\times E$ is a partial order of \emph{causal
    dependency},
\item $\#\subseteq E\times E$ is a binary, irreflexive, symmetric
  relation of \emph{conflict},
\item $\downarrow \!e:=\{ e'\in E: e'\le e\}$ is finite for any
  $e\in E$,
\item $e\# e'$ and $e'\le e''$ imply $e\# e''$.
\end{itemize}

Two events $e',e''$ are \emph{concurrent} (notation $e'\| e''$) if
they are order-incomparable and they are not in conflict.
A \emph{configuration} of an event structure ${\E}$ is any finite
subset $c\subset E$ of events which is \emph{conflict-free}
($e,e'\in c$ implies that $e,e'$ are not in conflict) and
\emph{downward-closed} ($e\in c$ and $e'\le e$ implies that
$e'\in c$). Notice that $\varnothing$ is always a configuration and
that $\downarrow \!e$ and $\downarrow \!e\setminus \{ e\}$ are
configurations for any $e\in E$.  The \emph{domain} of $\E$ is the set
$\cD({\E})$ of all configurations of ${\E}$ ordered by inclusion;
$(c',c)$ is a (directed) edge of the Hasse diagram of the poset
$({\cD}({\E}),\subseteq)$ if and only if $c=c'\cup \{ e\}$ for an
event $e\in E\setminus c$. The domain $\cD(\E)$ can be endowed with
the \emph{Hamming distance} $d(c,c')=|c\Delta c'|$ between any
configurations $c$ and $c'$. From the following result, the Hamming
distance coincides with the graph-distance.
Barth\'elemy and Constantin~\cite{BaCo}
established the following bijection between event structures and
pointed median graphs:

\begin{theorem}[\!\!\cite{BaCo}]\label{median_domain}
  The (undirected) Hasse diagram of the domain $(\cD(\E),\subseteq)$
  of an event structure $\E=(E,\le, \#)$ is
  median. Conversely, for any median graph $G$ and any basepoint $v$
  of $G$, the pointed median graph $G_v$ is the Hasse diagram of the
  domain of an event structure $\E_v$.
\end{theorem}

We briefly recall the construction of the event structure $\E_v$.
Consider a median graph $G$ and an arbitrary basepoint $v$. The events
of the event structure $\E_{v}$ are the hyperplanes of the cube
complex $\G$ (or the $\Theta$-classes of $G$). Two hyperplanes $H$ and
$H'$ define concurrent events if and only if they cross (i.e., there
exist a square with two opposite edges in one $\Theta$-class and other
two opposite edges in the second $\Theta$-class). The hyperplanes $H$
and $H'$ are in relation $H\leq H'$ if and only if $H=H'$ or $H$
separates $H'$ from $v$.  Finally, the events defined by $H$ and $H'$
are in conflict if and only if $H$ and $H'$ do not cross and neither
separates the other from $v$.

\begin{example}\label{example-event-structure1}
  The pointed median graph $G$ described in Fig.~\ref{fig:ExStrEvt} is
  the domain of the event structure $\E=(E,\leq, \# )$. The seven
  events $e_1,\ldots, e_7$ of $E$ correspond to the seven
  $\Theta$-classes of $G$. The causal dependency is defined by
  $e_1 \leq e_3, e_5, e_6, e_7$; $e_2 \leq e_4, e_5, e_6, e_7$;
  $e_3, e_4, e_5 \leq e_6, e_7$.  The events $e_6$ and $e_7$ are in
  conflict and all remaining pairs of events are concurrent.
\end{example}	

\begin{figure}[h]
  \centering
  {\includegraphics[scale=0.37]{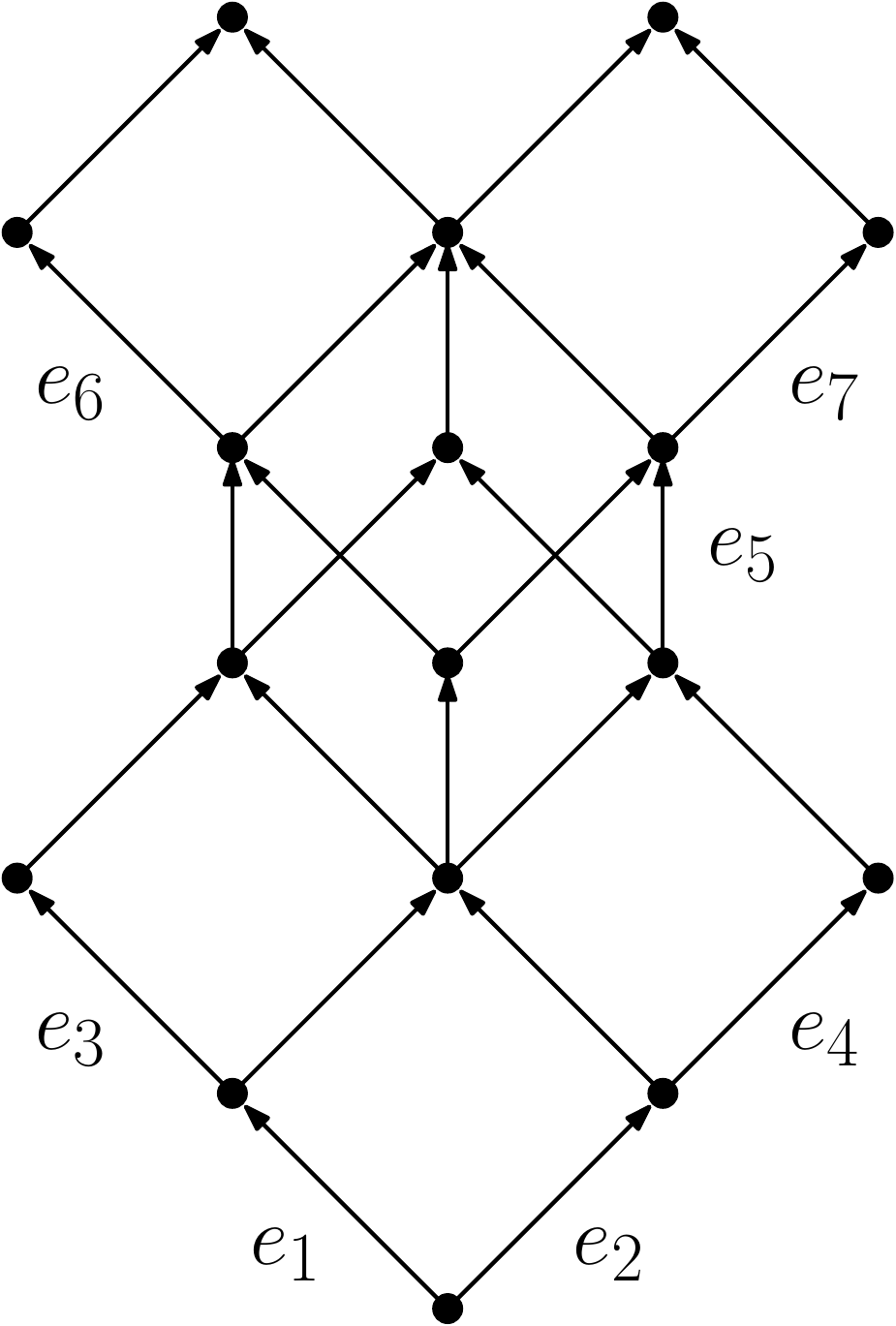}
    \caption{The domain of the event structure from
      Example~\ref{example-event-structure1}.}\label{fig:ExStrEvt}}
\end{figure}

Consequently, event structures encode median graphs and this
representation is much more compact than the standard one using
vertices and edges.  For example, the hypercube of dimension $d$ is
the domain of the event structure with $d$ events that are pairwise
concurrent.

\subsubsection{The median problem in event structures}
Let ${\E}=(E,\le, \#)$ be a finite event structure.  The input is a
set $C=\{ c_1,\ldots,c_k\}$ of configurations of $\E$ and their
weights $w_1,\ldots,w_k$, where each $c_i$ is given by the list of
events belonging to $c_i$. The goal of the \emph{median problem in the
  event structure} $\E$ is to compute a configuration $c$ minimizing
the function $F_w(c)=\sum_{i=1}^k w_i d(c,c_i)$, where $d(c,c')$ is
the Hamming distance between $c$ and $c'$. Consider also a special
case of this median problem, in which $C$ is the set of \emph{all}
configurations of $\E$ and the input is the event structure $\E$,
i.e., the graphs $(E,\le)$ and $(E,\#)$. We call this problem the
\emph{compact median problem}.

\subsubsection{2-SAT formulas}
A \emph{2-SAT formula} on variables $x_1,\ldots,x_n$ is a formula
$\varphi$ in conjunctive normal form with two literals per clause,
i.e., a set of clauses of the form $(u \vee v)$, where each of the two
literals $u,v$ is either a positive literal $x_i$ or a negative
literal $\neg x_i$. A \emph{solution} for $\varphi$ is an assignment
$S$ of variables to 0 or 1 that satisfies all clauses.  The
\emph{solution set} ${\mathcal S}(\varphi)$ of $\varphi$ is the set of
all solutions of $\varphi$.  We consider each solution set as a subset
of vertices of the $n$-dimensional hypercube $Q_n$.  A subset $Y$ of
vertices of the hypercube $Q_n$ (viewed as a median graph) is called
\emph{median-stable} if the median of each triplet $x,y,z\in Y$ also
belongs to $Y$.

\begin{proposition}[\!\!\cite{Scha,MuSch}]\label{median-2SAT}
  Median-stable sets are exactly the solution sets of 2-SAT formulas.
\end{proposition}

A variable of a 2-SAT formula $\varphi$ is \emph{trivial} if it has
the same value in each solution.  Two nontrivial variables $x_i$ and
$x_j$ in $\varphi$ are \emph{equivalent} if either $x_i=x_j$ in all
solutions or $x_i=\neg x_j$ in all solutions.  Here is a
characterization of 2-SAT formulas corresponding to median graphs:

\begin{proposition}[\!\!{\cite[Corollary 3.34]{Fe}}]\label{2sat=median-graph}
  The solution set ${\mathcal S}(\varphi)$ of a 2-SAT formula
  $\varphi$ induces a median graph if and only if $\varphi$ has no
  equivalent variables.
\end{proposition}

\subsection{A direct correspondence between event structures and 2-SAT formulas}
We provide a canonical correspondence between event structures and
2-SAT formulas (which may be useful also for other purposes).

By Theorem~\ref{median_domain} and
Propositions~\ref{median-2SAT},\ref{2sat=median-graph} there is a
bijection between the domains of event structures and pointed median
graphs and a bijection between (unpointed) median graphs and 2-SAT
formulas not containing equivalent variables.  Since $\varnothing$ and
$\downarrow \!e, e\in E$ are configurations, their characteristic
vectors must be solutions of the associated 2-SAT formula.  This can
be ensured by requiring that the 2-SAT formula does not contain
clauses of the form $(x_i\vee x_j)$.

Let $\E=(E, \leq, \#)$ be an event structure with
$E=\{ e_1,\ldots,e_n\}$. We associate to $\E$ a 2-SAT formula
$\varphi_{\E}$ on $n$ variables $x_1,\ldots,x_n$. For each pair of
events such that $e_i \leq e_j$ we define the clause
$(x_i \vee \neg x_j)$ and for each pair of events such that
$e_i \# e_j$ we assign the clause $(\neg x_i \vee \neg x_j)$. Next for
each subset $c$ of $E$ we denote by $S_c$ its characteristic vector.

\begin{proposition}\label{event-struc->2SAT}
  ${\mathcal S}(\varphi_{\E})$ coincides with $\cD(\E)$.
\end{proposition}

\begin{proof}
Let $c\subseteq E$ and $c\notin \cD(\E)$, i.e., $c$ is either not
  downward-closed or not conflict-free. In the first case, there exist
  two events $e_i$ and $e_j$ such that $e_i \leq e_j$ and
  $e_j\in c, e_i\notin c$. This implies that $\varphi_{\E}$ contains
  the clause $(\neg x_j \vee x_i)$. Since $S_c(x_j)=1$ and
  $S_c(x_i)=0$, $S_c$ is not a solution of $\varphi_{\E}$. In the
  second case, there exist two events $e_i,e_j\in c$ such that
  $e_i \# e_j$. This implies that $\varphi_{\E}$ contains the clause
  $(\neg x_i \vee \neg x_j)$. Since $S_c(x_i)=S_c(x_j)=1$, $S_c$ is
  not a solution of $\varphi_{\E}$.
		
  Conversely, suppose that $S_c$ is not a solution of
  $\varphi_{\E}$. Recall that $\varphi_{\E}$ contains only clauses of
  the form $(\neg x_i \vee x_j)$ or $(\neg x_i \vee \neg x_j)$. If a
  clause $(\neg x_i \vee x_j)$ of $\varphi_{\E}$ is false, then
  $S_c(x_i)=1$ and $S_c(x_j)=0$. Thus, the events $e_i$ and $e_j$ are
  such that $e_j \leq e_i$ and $e_i\in c$ and $e_j\notin c$.
  Consequently, $c$ is not a configuration of $\E$. Similarly, if a
  clause $(\neg x_i \vee \neg x_j)$ in $\varphi_{\E}$ is false, then
  $S_c(x_i)=S_c(x_j)=1$. Thus $c$ contains two events $e_i$ and $e_j$
  such that $e_i \# e_j$, whence $c$ is not a configuration. This
  shows that ${\mathcal S}(\varphi_{\E})$ and $\cD(\E)$ coincide.
\end{proof}

Let $\varphi$ be a 2-SAT formula on variables $x_1,\ldots, x_n$ not
containing trivial and equivalent variables and clauses of the form
$(x_i\vee x_j)$. We associate to $\varphi$ an event structure
$\E_{\varphi}=(E, \leq, \#)$ consisting of a set
$E=\{ e_1,\ldots,e_n\}$ and the binary relations $\leq$ and $\#$
defined as follows.  First we define two binary relations $\#_0$ and
$\leq_0$: for $e_i,e_j$ we set $e_i \#_0 e_j$ if and only if $\varphi$
contains the clause $(\neg x_i \vee \neg x_j)$ and we set
$e_i \leq_0 e_j$ if and only if $\varphi$ contains the clause
$(x_i \vee \neg x_j)$.  Let $\leq$ be the transitive and reflexive
closure of $\leq_0$. Let also $\#$ be the relation obtained by setting
$e_j \# e_{\ell}$ for each quadruplet $e_i, e_j, e_k, e_{\ell}\in E$
such that $e_i \leq e_j,e_k \leq e_{\ell}$, and $e_i \#_0 e_k$. Note
that $\#_0\subseteq \#$ and that $\#$ satisfies the last axiom of
event structures.

\begin{proposition}\label{2SAT->event-struc}
  $\E_{\varphi}=(E, \leq, \#)$ is an event structure and
  $\cD(\E_{\varphi})$ coincides with ${\mathcal S}(\varphi)$.
\end{proposition}

\begin{proof}
  In view of previous conclusions, to show that $\E_{\varphi}$ is an
  event structure it remains to prove that $\leq$ is
  antisymmetric. Suppose by way of contradiction that there exist
  $e_i,e_j\in E$ such that $e_i \leq e_j$ and $e_j \leq e_i$.  By
  definition of $\leq$, there exist $e_1, \ldots, e_p\in E$ and
  $e_{p+1}, e_{p+2}, \ldots, e_q\in E$ such that
  $e_i \leq_0 e_1 \leq_0 e_2 \leq_0 \cdots \leq_0 e_p \leq_0 e_j$ and
  $e_j \leq e_{p+1} \leq_0 e_{p+2} \leq_0 \cdots \leq_0 e_q \leq_0
  e_i$.  Consequently, the formula $\varphi$ contains the clauses
  $(\neg x_i \vee x_1), (\neg x_1 \vee x_2), \ldots, (\neg x_p \vee
  x_j),(\neg x_j \vee x_{p+1}), \ldots, (\neg x_q \vee x_i)$.  Thus,
  the variables $x_1, \ldots, x_q, x_i, x_j$ are equivalent, which is
  impossible because $\varphi$ is a 2-SAT formula without trivial and
  equivalent variables.
		
  To prove the second assertion, let $c=\{ e_1, \ldots e_p\}$ be a
  subset of $E$ which is not a configuration of $\E_{\varphi}$, i.e.,
  either $c$ is not downward-closed or is not conflict-free. First,
  suppose that $c$ contains an event $e_j$ such that there is an event
  $e_i\notin c$ with $e_i \leq e_j$.  Since $\leq$ is the transitive
  and reflexive closure of $\leq_0$, there exists a pair of events
  $e_k,e_{\ell}$ such that $e_{\ell}$ belongs to $c$, $e_k$ does not
  belong to $c$ and $e_k \leq_0 e_{\ell}$.  Thus $\varphi$ contains
  the clause $(\neg x_{\ell} \vee x_k)$. Since $S_c(x_{\ell})=1$ and
  $S_c(x_k)=0$, the assignment $S_c$ is not a solution of
  $\varphi$. Therefore, we can suppose now that $c$ is
  downward-closed.  Second, suppose that $c$ contains two events $e_i$
  and $e_j$ such that $e_i \# e_j$.  By definition of $\#$ and $\#_0$,
  there exists a pair of events $e_k \leq e_i$ and $e_{\ell} \leq e_j$
  such that $e_k \#_0 e_{\ell}$. Since $c$ is downward-closed and
  contains both $e_i$ and $e_j$, necessarily both $e_k$ and $e_{\ell}$
  belong to $c$. Thus, $S_c(x_k)=S_c(x_{\ell})=1$. But since
  $e_k \#_0 e_{\ell}$, $\varphi$ contains the clause
  $(\neg x_k \vee \neg x_{\ell})$, which is not satisfied by
  $S_c$. Consequently, if $c$ is not a configuration of $\E$, then
  $S_c$ is not a solution of $\varphi$.
		
  Conversely, suppose that $S$ is a assignment which is not a solution
  of $\varphi$. Let $c\subseteq E$ such that $S_c=S$.  Since $\varphi$
  does not contain clauses of the form $(x_i\vee x_j)$, this implies
  that $\varphi$ either contains a clause $(\neg x_i \vee x_j)$ such
  that $S_c(x_i)=1$ and $S_c(x_j)=0$, or $\varphi$ contains a clause
  $(\neg x_i \vee \neg x_j)$ such that $S_c(x_i)=S_c(x_j)=1$.  If
  $\varphi$ contains $(\neg x_i \vee x_j)$ with $S_c(x_i)=1$ and
  $S_c(x_j)=0$, then $e_j \leq e_i$ in $\E_{\varphi}$. Consequently,
  the corresponding subset $c$ of events is not a configuration of
  $\E_{\varphi}$ because $c$ contains $e_i$ but not $e_j$. If
  $\varphi$ contains $(\neg x_i \vee \neg x_j)$ with
  $S_c(x_i)=S_c(x_j)=1$, then $e_i \# e_j$ and again $c$ is not a
  configuration of $\E_{\varphi}$ because $c$ contains two conflicting
  events $e_i$ and $e_j$.
\end{proof}

\subsection{The median problem in event structures}\label{median_event}

In this subsection we show that the median problem in event structures
can be solved in linear time in the size of the input. We also show
that a diametral pair of median configurations can be computed in
linear time.  On the other hand, we show that the compact median
problem is hard.

\subsubsection{An algorithm for the median problem in event structures}
Let $\E=(E, \leq, \#)$ be an event structure with
$E=\{ e_1,\ldots,e_n\}$. Let $C=\{ c_1,\ldots,c_k\}$ be a set of
configurations of $\E$ and $w_1,\ldots,w_k$ be their weights. Let
$c^*$ be a subset of $E$ defined by the majority rule in the hypercube
$Q_n=\{ 0,1\}^E$.  Namely, $c^*$ consists of all events $e_i$ such
that the weight of all configurations of $C$ containing $e_i$ is
strictly larger than the total weight of all configurations not
containing $e_i$:
$c^*=\{ e_i\in E: \sum_{j: e_i\in c_j} w_j>\sum_{j: e_i\notin c_j}
w_j\}$. We assert that $c^*$ is a configuration of $\E$ and that $c^*$
minimizes the median function $F_w(c)=\sum_{i=1}^k w_i d(c,c_i)$.
Since $\cD(\E)$ is an isometric subgraph of $Q_n$, for each
$c\in \cD(\E)$, the values of the median function $F_w(c)$ in
$\cD(\E)$ and $Q_n$ are the same.

Since $Q_n$ is a median graph, by the majority rule
(Proposition~\ref{majority}), the median set of $Q_n$ is the
intersection of all majoritary halfspaces of $Q_n$. Each pair
$H'_i,H''_i$ of complementary halfspaces of $Q_n$ correspond to an
event $e_i$ of $\E$: one halfspace $H'_i$ consists of all
$c\subseteq E$ containing $e_i$ and its complement $H''_i$ consists of
all $c\subseteq E$ not containing $e_i$. If $w(H'_i)>w(H''_i)$, then
$H'_i$ is majoritary, which means that the weight of all
configurations of $C$ containing $e_i$ is strictly larger than one
half of the total weight. By definition of $c^*$, $e_i$ belongs to
$c^*$, i.e., $c^*$ (viewed as a characteristic vector) is a vertex of
$H'_i$. Similarly, if $w(H''_i)>w(H'_i)$, then $H''_i$ is majoritary,
which means that the weight of all configurations of $C$ not
containing $e_i$ is strictly larger than one half. By definition of
$c^*$, $e_i$ does not belong to $c^*$, i.e., $c^*$ is a vertex of
$H''_i$.  Consequently, $c^*$ is a vertex of the hypercube $Q_n$
minimizing the function $F_w(c)=\sum_{i=1}^k w_i d(c,c_i)$ over all
$c\subseteq E$. Since the minimum of this function taken over
$\cD(\E)$ cannot be smaller than this minimum, to finish the proof it
remains to show that $c^*$ is a configuration of $\E$. Suppose that
$e_i\in c^*$ and $e_j\leq e_i$. Since each configurations of $\E$ is
downward-closed, all configurations of $C$ containing $e_i$ also
contain $e_j$. Thus the weight of all configurations of $C$ containing
$e_j$ is strictly larger than the weight of all configurations not
containing $e_j$, whence $e_j\in c^*$. Now suppose by way of
contradiction that $e_i,e_j\in c^*$ and $e_i\# e_j$. Since
$e_i,e_j\in c^*$, the weight of all configurations of $C$ containing
$e_i$ is strictly larger than one half of the total weight and the
weight of all configurations of $C$ containing $e_j$ is also strictly
larger than one half of the total weight. Therefore, in $C$ we must
find a configuration containing both $e_i$ and $e_j$, which is
impossible because $e_i\# e_j$. Consequently, we obtain the following
result:

\begin{proposition}
  A median configuration $c^*$ of any set $C=\{ c_1,\ldots,c_k\}$ of
  configurations of an event structure $\E$ can be computed in linear
  time in the size $O(\sum_{i=1}^k |c_i|)$ of the input.
\end{proposition}

\begin{remark}
  We mentioned in the Introduction that the space of trees with a
  fixed set of $n$ leaves is a CAT(0) cube complex~\cite{BiHoVo}.  The
  vertices of this complex are the so-called $n$-trees and is was
  known since 1981 that the set of all $n$-trees is a median
  semilattice~\cite{MaMcM}, thus a median graph. Let
  $E=\{ e_1,\ldots, e_n\}$. An $n$-\emph{tree} $T$ is a collection of
  subsets of $E$ satisfying the following conditions: (1)
  $E\in T, \varnothing \notin T$, (2) $\{ e_i\}\in T$ for any
  $e_i\in E$, (3) $A\cap B\in \{ \varnothing, A,B\}$ for any
  $A,B\in T$. Any set $A\in T$ is called a \emph{cluster}. This name
  is justified by the fact that $n$-trees are exactly the collections
  of clusters occurring in hierarchical clustering: any two clusters
  either are disjoint or one is contained in another one. Barth\'elemy
  and McMorris~\cite{BaMcM} considered the median problem for
  $n$-trees, where the input consists of the $n$-trees
  $T_1,\ldots, T_k$ on $E$ and the goal is to compute an $n$-tree $T$
  minimizing $\sum_{i=1}^k d(T,T_i)$, where $d(T,T')=|T\Delta T'|$ is
  the number of clusters in $T$ but not in $T'$ plus the number of
  clusters in $T'$ but not in $T$. Since the space of all $n$-trees is
  a median semilattice, the authors of~\cite{BaMcM} deduced that the
  majority $n$-tree $T^*$ is a median $n$-tree; the \emph{majority
    $n$-tree} $T^*$ consists of all clusters included in strictly more
  than one half of the $n$-trees $T_1,\ldots, T_k$. This can be viewed
  as another compact formulation of the median problem in an
  implicitly defined median graph, where the input is given by the
  $n$-trees $T_1,\ldots, T_k$.
\end{remark}

\begin{remark}\label{median-2sat}
  Due to the correspondence between event structures and 2-SAT
  formulas establishes in Propositions~\ref{event-struc->2SAT}
  and~\ref{2SAT->event-struc}, we can define a similar median problem
  for a set of solutions $S_{c_1},\ldots,S_{c_k}$ of the 2-SAT formula
  $\varphi_{\E}$ and to search for a solution
  $S_c\in {\mathcal S}(\varphi_{\E})$ minimizing
  $\sum_{i=1}^k w_id(S_c,S_{c_i})$. From our bijections we deduce that
  $S_{c^*}$ belongs to ${\mathcal S}(\varphi_{\E})$ and therefore is
  an optimal solution.
\end{remark}

\subsubsection{Computing a diametral pair of median configurations} 
We know from~\cite{BaBa} that the median set of a median graph
coincides with the interval between two diametral pairs of its
vertices. In Proposition~\ref{interval} we gave a different proof of
this result and in Corollary~\ref{cor-interval} we showed how to find
such a diametral pair in linear time.  Now we will show how to compute
a diametral pair $\{ c'_*,c''_*\}$ of median configurations in linear
time in the size of the input (the description of the event structure
and of the set of configurations).

\begin{remark}
  Similarly to the median problem in the cube complexes associated to
  median graphs, we cannot explicitly return all median
  configurations, because one can have an exponential number of such
  optimal solutions.
\end{remark} 

For the moment, we will suppose that $\{ c',c''\}$ is a diametral pair
of the median set of configurations (which exists by the result
of~\cite{BaBa}). Recall also that in the previous subsection we
defined a canonical median configuration $c^*$. Similarly to the
classification of $\Theta$-classes of a median graph, we can classify
the events of $\E$ in three classes: an event $e\in E$ is called (i)
\emph{majoritary} if $e$ belongs to $c^*$, (ii) \emph{minoritary} if
the halfspace defined by $e$ and containing $c^*$ is a minoritary
halfspace, and (iii) \emph{egalitarian} if the two halfspaces defined
by $e$ have the same weight.  We denote by $E_=$ the set of all
egalitarian events. We start with several simple assertions.

\begin{lemma}\label{egalitarian_distance}
  The distance $d(c',c'')$ between $c'$ and $c''$ in the median graph
  $\cD(\E)$ equals to the number of egalitarian events.
\end{lemma}

\begin{proof}
  By Proposition~\ref{majority}, no majoritary or minoritary event of
  $\E$ separate two configurations of the median set. Therefore, any
  event corresponding to a $\Theta$-class separating $c'$ and $c''$ is
  egalitarian; this establishes that $d(c',c'')$ is not larger than
  $|E_=|$.  Conversely, we assert that any event $e\in E_=$ separates
  $c'$ and $c''$. By Proposition~\ref{majority}, the two halfspaces
  defined by $e$ both intersect the medians set. If $c',c''$ are not
  separated by these halfspaces, they necessarily they both belong to
  the same halfspace and some median configuration $c$ belongs to the
  complementary halfspace. Since $c\in I(c',c'')$, we obtain a
  contradiction with the convexity of halfspaces.  This proves that
  any egalitarian event separates $c'$ and $c''$.
\end{proof}

\begin{lemma}\label{median_c*}
  The median configuration $c^*$ is the gate of the empty
  configuration $c_{\varnothing}=\varnothing$ in the interval
  $I(c',c'')$. In particular, $c^*$ is the median of the triplet
  $c_{\varnothing},c',$ and $c''$.
\end{lemma}

\begin{proof}
  Suppose by way of contradiction that $c\ne c^*$ is the gate of
  $c_{\varnothing}$ in $I(c',c'')$. This implies that
  $c\in I(c_{\varnothing},c^*)$, i.e., $c^*$ is the union of $c$ and
  the events separating $c$ and $c^*$. Since $c,c^*\in I(c',c'')$, by
  Lemma~\ref{egalitarian_distance} any event separating $c$ and $c^*$
  is an egalitarian event. This contradicts the definition of $c^*$:
  by its definition, $c^*$ contains only majoritary events.
\end{proof}

By Lemma~\ref{median_c*},
$c^*\in I(c_{\varnothing},c')\cap I(c_{\varnothing},c'')$ and we
conclude that $c'=c^* \cup A$ and $c''=c^*\cup B$, where $A$ and $B$
are sets of egalitarian events. By Lemma~\ref{egalitarian_distance},
$d(c',c'')=|E_=|$ and since it coincides with the Hamming distance
$|c'\Delta c''|=|A\Delta B|$, the sets $A$ and $B$ must constitute a
partition of $E_{=}$. Since $c'=c^*\cup A$ and $c''=c^*\cup B$ are
configurations, we conclude that the sets $c^*\cup A, c^*\cup B$ are
conflict-free and downward-closed.  Therefore, the events of $c^*$ are
not in conflict with the events of $E_{=}=A\cup B$.

On the set $E_=$ we define the following binary relation $R_0$: for
$e_1,e_2\in E_=$ we set $e_1R_0 e_2$ if $e_1 \leq e_2$ or
$e_2 \leq e_1$. Let $R$ be the transitive closure of the relation
$R_0$. Observe that the equivalence classes of $R$ are the connected
components of the graph obtained by forgetting the orientation of the
Hasse diagram of $(E_{=},\leq)$. Now define the following conflict
graph $\Gamma$: the vertices of $\Gamma$ are the equivalence classes
of $R$ and two such classes $C'$ and $C''$ are linked by an edge in
$\Gamma$ if and only if there exists an event $e'\in C'$ and an event
$e''\in C''$ such that $e'\# e''$.

\begin{lemma}\label{Gamma-bipartite}
  Any equivalence class $C$ of the relation $R$ is
  conflict-free. Consequently, the conflict graph $\Gamma$ is
  bipartite.
\end{lemma} 

\begin{proof}
  Let $A,B$ be a bipartition of $E_=$ such that
  $c'=c^*\cup A, c''=c^*\cup B$ is a diametral pair of median
  configurations (that $c',c''$ have such a representation follows
  from the discussion after Lemma~\ref{median_c*}). Since the sets $A$
  and $B$ are conflict-free, it suffices to prove that $C$ is
  contained in $A$ or in $B$.

  Suppose by way of contradiction that there exists $e \in A\cap C$
  and $e' \in B\cap C$. By the definition of $R_0$, there exist events
  $e=e_0,e_1,\ldots,e_p,e_{p+1}=e'\in E_=$ such that
  $(e_0,e_1), (e_1,e_2),\ldots,(e_{p-1},e_p),(e_p,e_{p+1})\in
  R_0$. Since $A,B$ is a partition of $E_{=}$, there exists
  $(e_{j-1},e_j) \in R_0$ such that $e_{j-1} \in A \setminus B$ and
  $e_j \in B \setminus A$. Since $(e_{j-1},e_j) \in R_0$, either
  $e_{j-1}\leq e_j$ or $e_j \leq e_{j-1}$. Without loss of generality,
  assume the first. Consequently, since $c''=c^* \cup B$ is
  downward-closed and contains $e_j$, necessarily $e_{j-1}$ belongs to
  $c''$. Since $e_{j-1}$ is egalitarian and all events in $c^*$ are
  majoritary, necessarily $e_{j-1} \in B$, a contradiction.
Therefore, any equivalence class of $R$ is contained in $A$ or in
  $B$. Since $A$ and $B$ are conflict-free, any edge of the conflict
  graph $\Gamma$ must run between $A$ and $B$. Therefore, the
  equivalent classes of $R$ included in $A$ and those included in $B$
  define a bipartition of $\Gamma$ into two independent sets.
\end{proof}

\begin{lemma}\label{Gammma-bipartition}
  For the partition $A_*,B_*$ of $E_=$ induced by any bipartition
  $Q'_*,Q''_*$ of $\Gamma$ into two independent sets,
  $c'_*=c^*\cup A_*$ and $c''_*=c^*\cup B_*$ is a diametral pair of
  median configurations.
\end{lemma}

\begin{proof}
  Let $A_*$ be the union of all equivalence classes of $R$ contained
  in $Q'_*$ and let $B_*$ be the union of all equivalence classes of
  $R$ contained in $Q''_*$.  We assert that $c'_*=c^*\cup A_*$ and
  $c''_*=c^*\cup B_*$ are configurations of $\E$. We prove this
  assertion for $c'_*$. Since by Lemma~\ref{Gamma-bipartite} each
  equivalence class of $R$ is conflict-free and since $Q'_*$ is an
  independent set of $\Gamma$, the set $A_*$ is necessarily
  conflict-free. As we noticed above, no event of $c^*$ and of $E_=$
  are in conflict. Consequently, the set $c'_*=c^*\cup A_*$ is
  conflict-free.

  Now we show that $c'_*=c^*\cup A_*$ is downward-closed. Pick any
  event $e\in c'_*$ and any event $e'\leq e$.  If $e\in c^*$, then
  $e'\in c^*$ because we proved that $c^*$ is a configuration. Now
  suppose that $e\in A_*$.  Suppose by way of contradiction that
  $e'\notin c_*' = c_* \cup A_*$, i.e., either $e'$ is minoritary or
  $e'$ is egalitarian but belongs to
  $B_*$.
Since any configuration containing $e$ also contains $e'$, the total
  weight of the configurations containing $e'$ is at least the total
  weight of the configurations containing $e$. Since $e$ is
  egalitarian, either $e'$ is majoritary or $e'$ is egalitarian. In
  the first case, $e' \in c^*$ by the definition of $c^*$. In the
  second case, $e$ and $e'$ belong to the same equivalence class of
  $R$ and thus they both belong to $A_*$ or they both belong to $B_*$.

  Consequently, $c'_*=c^*\cup A_*$ and $c''_*=c^*\cup B_*$ are
  configurations of $\E$. We assert that both $c'_*$ and $c''_*$ are
  median configurations. Clearly, both $c'_*$ and $c''_*$ are
  contained in all halfspaces $H'_e$ for all majoritary events
  $e\in c^*$. On the other hand, $c'_*$ and $c''_*$ are contained in
  the halfspaces $H''_e$ for all minoritary events $e$.  Since in both
  cases those halfspaces have weight strictly larger than one half of
  the total weight, $c'_*$ and $c''_*$ belong to all majoritary
  halfspaces, thus by Proposition~\ref{majority} they are median
  configurations.  Finally, since $d(c'_*,c''_*)=|A_*\cup B_*|=|E_=|$,
  we conclude that $c'_*,c''_*$ is a diametral pair of the set of
  median configurations.
\end{proof} 

From previous results, we obtain the following linear time algorithm
for computing a diametral pair $c'_*,c''_*$ of median
configurations. First we compute the set $E_{\maj}$ of majoritary
events and the set $E_=$ of egalitarian events. This can be done in
total $O(\sum_{i=1}^k |c_i|)$ time by traversing the lists of events
describing the set of configurations $c_1,\ldots,c_k$. Next, compute
the binary relation $R_0$ and its reflexive and transitive closure
$R$. As noticed above, this can be done by computing the connected
components of the subgraph induced by $E_{=}$ of the Hasse diagram of
$\leq$ where we forget the orientation. This can be done in linear
time in the size of $(E,\leq)$. To compute the graph $\Gamma$, one
just need to consider the conflict graph $(E,\#)$ and for any
$e_1, e_2 \in E_{=}$ such that $e_1 \# e_2$, we add an edge between
the equivalence classes of $e_1$ and $e_2$. This can be done in linear
time in the size of $(E,\#)$ and the size of the graph $\Gamma$ is
linear in the size of $\E$.

\begin{proposition}\label{median-interval-event-structure}
  A diametral pair $c'_*=c^*\cup A_*,c''_*=c^*\cup B_*$ of median
  configurations of any set $C=\{ c_1,\ldots,c_k\}$ of configurations
  of an event structure $\E$ can be computed in
  $O(|\E|+\sum_{i=1}^k |c_i|)$ time.
\end{proposition} 

\subsubsection{The compact median problem is $\#$P-hard}\label{reduction}

An analogue of compact median problem for 2-SAT formulas was already studied by Feder~\cite{Fe}, who proved the following result:

\begin{proposition}[{\!\!\cite[Lemma 3.54]{Fe}}]\label{feder-median}
  For a 2-SAT formula $\varphi$ without trivial and equivalent
  variables, the problem of finding the median of the median graph
  ${\mathcal S}(\varphi)$ is $\#$P-hard.
\end{proposition}

The literals of any satisfiable 2-SAT formula $\varphi$ can be renamed
to transform $\varphi$ into an equivalent formula not containing
clauses $(x_i\vee x_j)$. Thus Proposition~\ref{feder-median} holds for
2-SAT formulas without trivial and equivalent variables and clauses
$(x_i\vee x_j)$. Since the size of the event structure $\E_{\varphi}$
in Proposition~\ref{2SAT->event-struc} is quadratic in the size of the
2-SAT formula $\varphi$, we obtain the following result from
Propositions~\ref{2SAT->event-struc},~\ref{feder-median}, and
Remark~\ref{median-2sat}:

\begin{proposition}\label{compact-median-problem}
  The compact median problem in an event structure $\E$ is $\#$P-hard.
\end{proposition}

\subsection*{Acknowledgements}
We would like to thank Florent Capelli, Nadia Creignou, and Yann
Strozecki for providing us with the proof of
Proposition~\ref{feder-median} much before we discovered this result
in~\cite{Fe}. The work on this paper was supported by ANR project
DISTANCIA (ANR-17-CE40-0015).

\bibliographystyle{amsplain}

\providecommand{\bysame}{\leavevmode\hbox to3em{\hrulefill}\thinspace}
\providecommand{\MR}{\relax\ifhmode\unskip\space\fi MR }
\providecommand{\MRhref}[2]{%
  \href{http://www.ams.org/mathscinet-getitem?mr=#1}{#2}
}
\providecommand{\href}[2]{#2}

\appendix
\section{Proofs from Section~\ref{sec:properties}}
	
We present the proofs of all auxiliary results from
Section~\ref{sec:properties}. Some of those results are well known to
the people working in metric graph theory, other results were also
known to us, but it is difficult to give the references first
presenting them.
	
\begin{proof}[Proof of Lemma~\ref{quadrangle}]
  Let $x$ be the median of the triplet $u,v,w$. Then $x$ must be
  adjacent to $v,w$ and must have distance $k-1$ to $u$. If there
  exists yet another such vertex $x'$, then the triplet $u,v,w$ will
  have two medians (or alternatively, the vertices $z,v,w,x,x'$ will
  define a $K_{2,3}$).
\end{proof}

\begin{proof}[Proof of Lemma~\ref{cube}]
  Pick any three squares $xutv, utwy$, and $vtwz$ of $G$, pairwise
  intersecting in three edges and all three intersecting in a single
  vertex. Since $G$ is bipartite and $K_{2,3}$-free,
  $d(x,w)=d(y,v)=d(z,u)=3$ and $d(x,y)=d(y,z)=d(z,x)=2$. Therefore,
  the median of vertices $x,y,z$ is a new vertex $r$ adjacent to
  $x,y$, and $z$ and having distance 3 to $t$. Consequently, the 8
  vertices induce an isometric 3-cube of $G$.
\end{proof}
	
\begin{proof}[Proof of Lemma~\ref{convex-gated}]
  That gated sets $S$ are convex holds for all graphs (and metric
  spaces). Indeed, pick $x,y\in S$ and any $z\in I(x,y)$. Let $z'$ be
  the gate of $z$ in $S$.  Then $z'\in I(z,x)\cap I(z,y)$. Since
  $z\in I(x,y)$ this is possible only if $z'=z$, i.e., $z\in
  S$. Conversely, let $S$ be a convex subgraph of a median graph $G$
  and pick any vertex $x\notin S$.  Let $v$ be a closest to $x$ vertex
  of $S$. Pick any vertex $y\in S$ and let $u$ be the median of the
  triplet $x,y,v$. Since $u\in I(y,v)$ and $S$ is convex, we deduce
  that $u\in S$. Since $u\in I(x,v)$, from the choice of $v$ we have
  $u=v$. Thus $v\in I(x,y)$, i.e., $v$ is the gate of $x$ in $S$.
\end{proof}
	
The convexity of halfspaces of median graphs and of their boundaries
was first established in~\cite{Mu}. We give a different and shorter
proof of this result.
	
\begin{proof}[Proof of Lemma~\ref{halfspaces}]
  For an edge $uv$ of a median graph $G$, recall that
  $W(u,v)=\{ x\in V: d(u,x)<d(v,x)\}$ and
  $W(v,u)=\{ x\in V: d(v,x)<d(u,x)\}$. We assert that $W(u,v)$ and
  $W(v,u)$ are convex. We use the local characterization of convexity
  of median graphs (and more general classes of graphs,
  see~\cite{Ch_metric}): a connected subgraph $H$ of a median graph
  $G$ is convex if and only if $I(x,y)\subseteq H$ for any two
  vertices $x,y$ of $H$ with $d_H(x,y)=2$. Since both sets $W(u,v)$
  and $W(v,u)$ induce connected subgraphs, we can use this result.
  Pick $x,y\in W(u,v)$ such that $x$ and $y$ have a common neighbor
  $z$ in $W(u,v)$. Suppose by way of contradiction that there exists a
  vertex $t\sim x,y$ belonging to $W(v,u)$.  Then
  $d(u,x)=d(u,y)=d(v,t)=k$ and $d(v,x)=d(v,y)=d(u,t)=k+1$. Since $G$
  is bipartite, $d(u,z)$ is either $k-1$ or $k+1$. If $d(u,z)=k+1$, by
  the quadrangle condition we will find a vertex $s\sim x,y$ at
  distance $k-1$ from $u$. But then the vertices $x,y,z,s,t$ induce a
  forbidden $K_{2,3}$. Thus $d(u,z)=k-1$, i.e., $d(v,z)=k$. Therefore
  $z,t\in I(x,v)$, $x\sim z,t$, and by the quadrangle condition there
  exists a vertex $r\sim t,z$ with $d(r,v)=k-1$. But then again
  $x,y,z,t,r$ induce a forbidden $K_{2,3}$. This contradiction
  establishes that for each edge $uv$ of $G$, $W(u,v)$ and $W(v,u)$
  induce convex and thus gated subgraphs of $G$.
		
  Next, define another binary relation $\Psi$ on edges of $G$: for two
  edges $uv$ and $xy$ we write $uv\Psi xy$ if and only if
  $x\in W(u,v)$ and $y\in W(v,u)$.  It can be easily seen that the
  relation $\Psi$ is reflexive and symmetric. Next we will prove that
  $\Psi$ is transitive and that $\Psi$ and $\Theta$ coincide. For
  transitivity of $\Psi$ it suffices to show that if $uv\Psi xy$, then
  $W(u,v)=W(x,y)$. Suppose by way of contradiction that there exists a
  vertex $z\in W(u,v)\setminus W(x,y)$. This implies that
  $z\in W(y,x)$, i.e., $y\in I(x,z)$. Since $x,z\in W(u,v)$ and
  $y\in W(v,u)$, this contradicts the convexity of
  $W(u,v)$. Consequently, $\Psi$ is an equivalence relation.
		
  To conclude the proof of the lemma it remains to show that
  $\Theta=\Psi$. It is obvious that $\Theta_0\subseteq \Psi$. Since
  $\Theta$ is the transitive closure of $\Theta_0$, we conclude that
  $\Theta\subseteq \Psi$. To show the converse inclusion
  $\Psi\subseteq \Theta$, pick any two edges $uv$ and $xy$ with
  $uv\Psi xy$. We proceed by induction on $k=d(u,x)=d(v,y)$. Let $x'$
  be a neighbor of $x$ in $I(x,u)\subseteq W(u,v)$. Then
  $d(x',v)=d(y,v)=k$ and $d(x,v)=k+1$. By the quadrangle condition,
  there exists a vertex $y'\sim x',y$ at distance $k-1$ from
  $v$. Since $y'\in I(y,v)\subseteq W(v,u)$, we conclude that
  $uv\Psi x'y'$. Since $d(u,x')=d(v,y')=k-1$, by induction hypothesis
  $uv\Theta x'y'$. Since $x'y'$ and $xy$ are opposite edges of a
  square, $x'y'\Theta_0 xy$, yielding $uv\Theta xy$. This finishes the
  proof of Lemma~\ref{halfspaces}.
\end{proof}

\begin{proof}[Proof of Lemma~\ref{convex-int-halfspaces}]
  Let $S$ be a convex set of $G$ and let $S'$ be the intersection of
  all halfspaces containing $S$. If $S$ is a proper subset of $S'$,
  since $S'$ is convex, we can find two adjacent vertices $u\in S$ and
  $v\in S'\setminus S$.  Let $uv\in E_i$, say $u\in H'_i$ and
  $v\in H''_i$. From Lemma~\ref{halfspaces} it follows that
  $H'_i=W(u,v)$ and $H''_i=W(v,u)$. Since $G$ is bipartite and $S$ is
  convex, all vertices of $S$ must be closer to $u$ than to $v$.
  Consequently, $S\subset W(u,v)=H'_i$.  Since $v\notin H'_i$, we
  obtain a contradiction with the definition of $S'$.
\end{proof}

\begin{proof}[Proof of Lemma~\ref{crossing}]
  One direction is trivial. Consider now a vertex $v$ and neighbors
  $v_1, \ldots, v_k$ of $v$ such that for all distinct
  $1 \leq i,j \leq k$, the respective $\Theta$-classes $E_i$ and $E_j$
  of $vv_i$ and $vv_j$ are crossing. Since all cubes of $G$ are
  locally convex, they are convex and gated, if we prove that $v$ and
  any subset of $k'$ neighbors of $v$ belong to a cube of dimension
  $k'$, then this $k'$-cube is unique.

  We show by induction on $k$ that there exists a cube $Q$ containing
  $v,v_1,v_2, \ldots, v_k$. If $k = 2$, without loss of generality,
  assume that $v \in H''_1 \cap H''_2$. Since $E_1$ and $E_2$ are
  crossing, there exists $u \in H'_1\cap H'_2$. Observe that $v_1$ and
  $v_2$ are respectively the gates of $v$ on $H'_i$ and
  $H'_j$. Consequently, $d(u,v_1) = d(u,v_2) = d(u,v) - 1$ and by the
  quadrangle condition, there exists $v' \sim v_1, v_2$. Therefore
  there is a square $vv_1v'v_2$, establishing the claim.
		
  Suppose now that the assertion holds for any $k ' < k$. By applying
  the lemma when $k = 2$ to $v_1, v_i$, for any $1 <i\leq k$, we know
  that there exists $u_i\sim v_1,v_i$. By induction hypothesis, there
  exists a unique $(k-1)$-cube $R'$ containing $v, v_2, \ldots, v_k$
  and a unique $(k-1)$-cube $R''$ containing $v_1, u_2, \ldots,
  u_k$. We assert that $R' \cup R''$ is a $k$-cube of $G$. For any
  vertex $v_i$, $i=2,\dots k$, since $u_i$ is a neighbor of $v_i$, by
  induction hypothesis, there exists a cube containing $u_i$ and the
  $k-2$ neighbors of $v_i$ in $R'$ distinct from $v$.  This implies
  that there exists an isomorphism between the facet of $R''$
  containing $u_i$ and not $v_1$ and the facet of $R'$ containing
  $v_i$ and not containing $v$. Since $R'$ and $R''$ are convex, this
  defines an isomorphism between $R'$ and $R''$ which maps vertices of
  $R'$ to their neighbors in $R''$. Hence $R'\cup R''$ defines a
  $k$-cube of $G$, finishing the proof of Lemma~\ref{crossing}.
\end{proof}

\begin{proof}[Proof of Lemma~\ref{boundary}]
  The last argument of the proof of Lemma~\ref{halfspaces} implies
  that the boundaries $\partial H'_i, \partial H''_i$ are connected
  subgraphs of $G$.  Thus it suffices to show that they are locally
  convex. Pick $x',y'\in \partial H'_i$ having a common neighbor
  $z'\in \partial H'_i$. Let $x'',y'',z''$ be the neighbors of
  respectively $x',y',z'$ in $\partial H''_i$ (they are unique because
  $H''_i$ is gated). Pick any common neighbor $t$ of $x',y'$ different
  from $z'$. Since $H'_i$ is convex, $t$ belongs to $H'_i$. By the
  cube condition, there exists a vertex $s$ adjacent to
  $t,x'',y''$. But then obviously $s\in H''_i$, whence
  $t\in \partial H'_i$.
		
  Since $\partial H'_i, \partial H''_i$ are convex, any vertex of
  $\partial H'_i$ is adjacent to exactly one vertex of
  $\partial H''_i$ and, vice versa, any vertex of $\partial H''_i$ is
  adjacent to exactly one vertex of $\partial H'_i$. This defines a
  bijection between $\partial H'_i$ and $\partial H''_i$.  Now, if
  $u',v'\in \partial H'_i$ are adjacent to
  $u'',v''\in \partial H''_i$, respectively, and $u'$ and $v'$ are
  adjacent, then $u''$ and $v''$ also must be adjacent. Indeed, since
  $G$ is bipartite, if $u''\nsim v''$, then $d(u'',v'')=3$ and $u',v'$
  belong to $I(u'',v''),$ contrary to the convexity of
  $\partial H''_i$. This shows that $\partial H'_i$ and
  $\partial H''_i$ are isomorphic subgraphs of $G$.
\end{proof}

\begin{proof}[Proof of Lemma~\ref{peripheral}]
  We have to prove that any furthest from the basepoint $v_0$
  halfspace $H'_i$ of $G$ is peripheral.  Let $x$ be the gate of $v_0$
  in $H'_i$. Then $x\in \partial H'_i$ and $d(v_0,x)=d(v_0,H'_i)$.
  Suppose by way of contradiction that the boundary $\partial H'_i$ is
  a proper subset of $H'_i$ and let $v$ be a closest to $v_0$ vertex
  in $H'_i\setminus \partial H'_i$ which is adjacent to a vertex $u$
  of $\partial H'_i$ (such a vertex exists because $H'_i$ is convex
  and thus connected). Let $E_j$ be the $\Theta$-class of the edge
  $uv$. Since $u$ is the gate of $v$ in $\partial H'_i$,
  $u\in I(x,v)$. Since $x\in I(v_0,v)$, we conclude that there exists
  a shortest $(v_0,v)$-path $P(v_0,v)$ passing via $x$ and
  $u$. Consequently, $v_0,u\in H''_j$ and $v\in H'_j$.
		
  Let $y$ be the gate of $v_0$ in $H'_j$. Clearly,
  $y\in \partial H'_j$ and $d(v_0,y)=d(v_0,H'_j)$. From the choice of
  $H'_i$ as a furthest from $v_0$ halfspace, $d(v_0,y)\le d(v_0,x)$.
  Since $x$ is the gate of $v_0$ in $H'_i$, this implies that $y$
  cannot be located in $H'_i$. Therefore $y$ belongs to $H''_i$. Let
  $z$ be the gate of $y$ in $H'_i$ (and $\partial H'_i$) and note that
  $z\in I(y,v)$. Since $u$ is the gate of $v$ in $\partial H'_i$, we
  also have $u\in I(v,z)$. Consequently, $u\in I(v,y)$ and belongs to
  $H'_j$ since $H'_j$ is convex, which is impossible.  This shows that
  $H'_i$ is peripheral and finishes the proof of
  Lemma~\ref{peripheral}.
\end{proof}

\begin{proof}[Proof of Lemma~\ref{descendent_cube}]
  To prove the downward cube property, pick any vertex $v$ and let
  $u_1,\ldots,u_d$ be the parents of $v$, i.e., the neighbors of $v$
  in $I(v_0,v)$. By the quadrangle condition, for any distinct
  $u_i, u_j$, $1 \leq i,j \leq d$, there exists a vertex
  $u_{i,j} \sim u_i,u_j$. This shows that the $\Theta$-classes of
  $vu_i$ and $vu_j$ are crossing. By Lemma~\ref{crossing}, this
  implies that $v, u_1, \ldots, u_d$ belong to a unique cube.
\end{proof}
\end{document}